\documentclass[sn-vancouver]{sn-jnl}%
\pdfoutput=1

\normalbaroutside %

\usepackage{anyfontsize}
\usepackage[english]{babel}
\usepackage{amssymb}

\usepackage{tikz}
\usetikzlibrary{positioning,arrows,petri,calc,decorations.markings,arrows.meta,decorations.pathmorphing,patterns.meta,shapes}
\tikzset{place/.style={circle,thick,minimum size=4mm,draw,fill=mygrey!15!white},
transitionh/.style={rectangle,thick,fill=black,minimum width=6mm,inner ysep=1pt},
transition/.style={rectangle,thick,fill=black,minimum height=6mm,inner xsep=1pt},
arc/.style={->,>=stealth'}}
\pgfdeclarelayer{back}
\pgfsetlayers{back,main}
\makeatletter
\pgfkeys{%
  /tikz/on layer/.code={
    \pgfonlayer{#1}\begingroup
    \aftergroup\endpgfonlayer
    \aftergroup\endgroup
  },
  /tikz/node on layer/.code={
    \pgfonlayer{#1}\begingroup
    \expandafter\def\expandafter\tikz@node@finish\expandafter{\expandafter\endgroup\expandafter\endpgfonlayer\tikz@node@finish}%
  },
}
\tikzset{%
glow/.style={%
preaction={#1, draw, line cap=round, line join=round, line width=0.5pt, opacity=0.04, on layer=back,
preaction={#1, draw, line cap=round, line join=round, line width=1.0pt, opacity=0.04, on layer=back,
preaction={#1, draw, line cap=round, line join=round, line width=1.5pt, opacity=0.04, on layer=back,
preaction={#1, draw, line cap=round, line join=round, line width=2.0pt, opacity=0.04, on layer=back,
preaction={#1, draw, line cap=round, line join=round, line width=2.5pt, opacity=0.04, on layer=back,
preaction={#1, draw, line cap=round, line join=round, line width=3.0pt, opacity=0.04, on layer=back,
preaction={#1, draw, line cap=round, line join=round, line width=3.5pt, opacity=0.04, on layer=back,
preaction={#1, draw, line cap=round, line join=round, line width=4.0pt, opacity=0.04, on layer=back,
preaction={#1, draw, line cap=round, line join=round, line width=4.5pt, opacity=0.04, on layer=back,
preaction={#1, draw, line cap=round, line join=round, line width=5.0pt, opacity=0.04, on layer=back,
preaction={#1, draw, line cap=round, line join=round, line width=5.5pt, opacity=0.04, on layer=back,
preaction={#1, draw, line cap=round, line join=round, line width=6.0pt, opacity=0.04, on layer=back,
}}}}}}}}}}}}}}
\usepackage{pgfplots}
\pgfplotsset{compat=newest}
\pgfplotsset{plot coordinates/math parser=false}
\newlength\figureheight
\newlength\figurewidth 
\usepackage{cancel}

\usepackage{mathtools}
\usepackage{etoolbox}
\usepackage[ruled,vlined,linesnumbered,algo2e]{algorithm2e}
\makeatletter
\patchcmd{\algocf@Vline}{\vrule}{\vrule\vspace{-.32em}}{}{}
\makeatother
\SetStartEndCondition{ }{ }{}%
\SetKw{KwTo}{to}\SetKwFor{For}{for}{\string do}{}%
\SetKwIF{If}{ElseIf}{Else}{if}{then}{elif}{else}{}%
\SetKwFor{While}{while}{\string do}{}%

\usepackage{caption}
\DeclareMathAlphabet{\mathcal}{OMS}{zplm}{m}{n}
\usepackage[autostyle, english = american]{csquotes}
\MakeOuterQuote{"}
\usepackage{xargs}
\usepackage{subfig}
\usepackage{stmaryrd} %

\usepackage{soul, color, xcolor}

\usepackage[numbers]{natbib}
\usepackage{enumitem}
\usepackage{nicematrix}
\usepackage{cleveref}
\usepackage{empheq} %
\usepackage{linegoal}

\definecolor{myblue}{RGB}{0, 101, 202}
\definecolor{mygreen}{RGB}{130, 180, 0}
\definecolor{myred}{RGB}{197, 14, 31}
\definecolor{mypurple}{RGB}{128, 0, 128}
\definecolor{myyellow}{RGB}{204, 204, 0}
\definecolor{mygrey}{RGB}{105, 105, 105}
\colorlet{LightGray}{mygrey!15!white}

\newcommand{\myblue}[1]{\textbf{\textcolor{myblue}{#1}}}

\newcommand{\dint}[1]{\left\llbracket#1\right\rrbracket} %

\newcommand{\D}{\mathcal{D}}
\renewcommand{\S}{\mathcal{S}}

\newcommand{\X}{\mathcal{X}}
\newcommand{\N}{\mathbb{N}}
\newcommand{\No}{\mathbb{N}_0}
\renewcommand{\Z}{\mathbb{Z}}

\renewcommand{\R}{\mathbb{R}}
\renewcommand{\Q}{\mathbb{Q}}
\newcommand{\Rmax}{{\R}_{\normalfont\fontsize{7pt}{11pt}\selectfont\mbox{max}}}
\newcommand{\Rmin}{{\R}_{\normalfont\fontsize{7pt}{11pt}\selectfont\mbox{min}}}
\newcommand{\Rbar}{\overline{\R}}

\newcommand{\zor}[1]{{#1}}

\newcommand{\graph}{\mathcal{G}}
\newcommand{\nonegset}{\Gamma} %
\newcommand{\places}{\mathcal{P}}
\newcommand{\transitions}{\mathcal{T}}
\newcommand{\marking}{m}
\newcommand{\arcs}{\mathcal{A}}
\newcommand{\nodes}{\mathcal{N}}
\newcommand{\myshift}{\textup{shift}}
\newcommand{\myweight}{\textup{weight}}
\newcommand{\mylshift}{\textup{Lshift}} %
\newcommand{\myrshift}{\textup{Rshift}} %
\newcommand{\myup}{\mbox{source}}
\newcommand{\mydown}{\mbox{target}}
\newcommand{\myheight}{\mbox{base}}
\newcommand{\mylen}{\mbox{len}}

\makeatletter
\newcommand{\splus}{%
  \DOTSB\mathop{\mathpalette\mattos@splus\relax}\slimits@
}
\newcommand\mattos@splus[2]{%
  \vcenter{\hbox{%
    \sbox\z@{$#1\oplus$}%
    \resizebox{!}{0.9\dimexpr\ht\z@+\dp\z@}{\raisebox{\depth}{$\m@th#1\boxplus$}}%
  }}%
  \vphantom{\oplus}%
}
\makeatother

\makeatletter
\newcommand{\stimes}{%
  \DOTSB\mathop{\mathpalette\mattos@stimes\relax}\slimits@
}
\newcommand\mattos@stimes[2]{%
  \vcenter{\hbox{%
    \sbox\z@{$#1\otimes$}%
    \resizebox{!}{0.9\dimexpr\ht\z@+\dp\z@}{\raisebox{\depth}{$\m@th#1\boxtimes$}}%
  }}%
  \vphantom{\otimes}%
}
\makeatother

\usepackage{pict2e}

\makeatletter
\newcommand*{\bigsplus}{\DOTSB\mathop{\mathpalette\big@boxplus\relax}\slimits@}

\newcommand{\big@boxplus}[2]{%
  \vcenter{%
    \m@th\bigbox@thickness{#1}%
    \sbox\z@{$#1\bigoplus$}%
    \dimen@=\ht\z@ \advance\dimen@\dp\z@
    \hbox{%
      \setlength{\unitlength}{\dimen@}%
      \begin{picture}(1,1)
      \polyline(0.1,0.1)(0.9,0.1)(0.9,0.9)(0.1,0.9)(0.1,0.1)(0.5,0.1)
      \polyline(0.5,0.1)(0.5,0.9)
      \polyline(0.1,0.5)(0.9,0.5)
      \end{picture}%
    }%
  }%
}

\newcommand{\bigbox@thickness}[1]{%
  \ifx#1\displaystyle
    \linethickness{0.2ex}%
  \else
    \ifx#1\textstyle
      \linethickness{0.16ex}%
    \else
      \ifx#1\scriptstyle
        \linethickness{0.12ex}%
      \else
        \linethickness{0.1ex}%
      \fi
    \fi
  \fi
}
\makeatother

\newcommand{\wP}{\mathsf{p}}
\newcommand{\wT}{\mathsf{t}}
\newcommand{\wN}{\mathsf{n}}

\renewcommand{\epsilon}{\varepsilon}

\newsavebox{\rightdiv}
\sbox\rightdiv{\tikz[anchor=south,baseline]{\footnotesize\node[inner sep=0pt,minimum height=.75em] at (0,0) (A){$\circ$};\node[inner sep=0pt,minimum height=.75em] at (0,0) {$/$};}}
\newsavebox{\rightmindiv}
\sbox\rightmindiv{\tikz[anchor=south,baseline]{\footnotesize\node[inner sep=0pt,minimum height=.75em] at (0,0) (A){$\bullet$};\node[inner sep=0pt,minimum height=.75em] at (0,0) {$/$};}}

\newsavebox{\leftdiv}
\sbox\leftdiv{\tikz[anchor=south,baseline]{\footnotesize\node[inner sep=0pt,minimum height=.75em] at (0,0) (A){$\circ$};\node[inner sep=0pt,minimum height=.75em] at (0,0) {$\setminus$};}}
\newsavebox{\leftmindiv}
\sbox\leftmindiv{\tikz[anchor=south,baseline]{\footnotesize\node[inner sep=0pt,minimum height=.75em] at (0,0) (A){$\bullet$};\node[inner sep=0pt,minimum height=.75em] at (0,0) {$\setminus$};}}

\newcommand{\eqtop}[1]{\mathrel{\overset{{\mbox{\normalfont\scriptsize #1}}}{=}}}

\newcommand{\ifftop}[1]{\mathrel{\overset{{\mbox{\normalfont\scriptsize #1}}}{\Leftrightarrow}}}

\newcommand{\ie}{i.e.,~}
\newcommand{\eg}{e.g.,~}

\newcommand{\qedlineend}{\makebox[0pt][l]{\makebox[\linegoal][r]{\qedhere}}}

\usepackage{amsthm}
\usepackage{thmtools}

\jyear{2024}%

\declaretheoremstyle[
  bodyfont=\normalfont,
  qed=\blacksquare,
]{theoremlike}
\declaretheoremstyle[
  bodyfont=\normalfont,
  qed=\lozenge,
]{definitionlike}
\declaretheorem[
  style=theoremlike
  ]{theorem,lemma,corollary,proposition}
\declaretheorem[%
  style=definitionlike%
  ]{definition,remark,example,assumption}

\begin{document}

\title[Infinite precedence graphs for consistency verification in P-TEGs]{Infinite precedence graphs for consistency verification in P-time event graphs}

\author*[1]{\fnm{Davide} \sur{Zorzenon}}\email{zorzenon@control.tu-berlin.de}

\author[1,2]{\fnm{J\"{o}rg} \sur{Raisch}}\email{raisch@control.tu-berlin.de}

\affil[1]{\orgdiv{Control Systems Group}, \orgname{Technische Universit\"{a}t Berlin}, \orgaddress{\country{Germany}}}

\affil[2]{\orgdiv{Science of Intelligence}, \orgname{Research Cluster of Excellence}, \orgaddress{\city{Berlin}, \country{Germany}}}

\abstract{
Precedence constraints are inequalities used to model time dependencies.
In 1958, Gallai proved that a finite system of precedence constraints admits solutions if and only if the corresponding precedence graph does not contain positive-weight circuits.
We show that this result extends naturally to the case of infinitely many constraints.
We then analyze two specific classes of infinite precedence graphs -- $\N$-periodic and ultimately periodic graphs -- and prove that the existence of solutions of their related constraints can be verified in strongly polynomial time.
The obtained algorithms find applications in P-time event graphs, which are a subclass of P-time Petri nets able to model production systems \zor{under cyclic schedules} where tasks need to be performed within given time windows.
}

\keywords{Precedence graphs, ultimately periodic graphs, P-time event graphs, max-plus algebra}

\maketitle

\section{Introduction}

In many production systems, ranging from the food industry to printed circuit boards manufacturing, the violation of temporal specifications can result in irreparable damage of the final product. 
When the logical sequence of operations to be performed is cyclically repeating, such systems can be modeled by P-time event graphs (P-TEGs).
P-TEGs are ordinary Petri nets where time intervals are associated to places, and each place has exactly one upstream and one downstream transition.

Failure to meet a temporal specification in the real system corresponds to a token remaining in a place of the P-TEG for longer or shorter than prescribed by the associated time-window constraint.
The primary goal of this paper is to study the consistency property in P-TEGs.
We say that a P-TEG is consistent if it admits an infinite sequence of transition firings that do not violate any constraint.
In a manufacturing system modeled by a consistent P-TEG, infinitely many products can, in principle, be processed without time-window constraint violations.
This is an interesting property in its own right, but the main motivation for studying consistency is that other interesting questions related to P-TEGs (and more complex systems) cannot be addressed without a thorough understanding of this property.
\zor{To make an analogy, just as} stability in standard dynamical systems can be analyzed only after confirming the existence of trajectories, the optimal throughput of a production system modeled by a P-TEG can be determined only once its consistency has been verified.

\subsection{Literature review}\label{su:literature_review}

The consistency verification problem has been considered by several authors.

Given a P-TEG with $n$ transitions and at most one initial token in each place,\footnote{Any P-TEG can be transformed into another, behaviorally equivalent one where this condition is met, at the cost of increasing the number of transitions. See \Cref{su:marking_transformation} for details.} the problem of checking the existence of acceptable trajectories of a given \emph{finite} length $h$ was proven to be solvable in time $O(hn^3)$ \zor{by Declerck} in~\cite{5628259}.
\zor{The same paper also showed that consistency can sometimes be verified or falsified in finite time by studying a certain sequence of matrices (presented in \Cref{le:formula_Pi}); if the sequence converges after finitely many iterations, indeed, then the corresponding P-TEG is consistent, whereas if some entries of the matrices in sequence diverge to $+\infty$, then the P-TEG is not consistent.
}

\zor{The paper} \cite{ZORZENON202219} \zor{introduced a property called} weak consistency.
A P-TEG is called weakly consistent if it admits trajectories of any finite length.
Interestingly, this property does not imply consistency, as some P-TEGs can admit finite trajectories of any length, but no infinite trajectory.\footnote{An example of weakly consistent but not consistent P-TEG is given in Figure~\ref{fi:P-TEG_example} for parameters $\alpha=-5$, $\beta=4$.}
\zor{For these P-TEGs, the entries from the sequence of matrices defined in \cite{5628259} neither converge nor diverge to $+\infty$ in finitely many iterations.
Therefore, the important result from \cite{5628259} is not strong enough to decide consistency in finite time.}
In~\cite{ZORZENON202219}, it was shown that weak consistency can be verified in strongly polynomial time $O(n^9)$.\footnote{\zor{The worst-case time complexity for checking weak consistency can be lowered to $O(n^7)$. Faster algorithms are not known by the authors of this paper.}}

The consistency verification problem was essentially solved for the case where upper bound constraints appear only in places with no initial tokens in~\cite[Corollary 2.3]{iteb2006control}.
Building on the same approach based on formal power series, in the PhD thesis~\cite{brunsch2014modeling}, the consistency problem was declared solved in its entirety.
However, the solution proposed is invalid due the presence of a technical error.\footnote{Using the notation of~\cite{brunsch2014modeling}, the dual product $\odot$ does \textit{not} distribute over the infimum $\wedge$ in the dioid $\mathcal{M}_{in}^{ax}\dint{\gamma,\delta}$. For a counterexample, one can check that, for $a=\gamma^3\delta^1$, $b=\gamma^5\delta^3$, $c=\gamma^1\delta^2\oplus \gamma^3\delta^5$, $(a\wedge b)\odot c = \gamma^6\delta^3\oplus \gamma^8\delta^6 \neq \gamma^6\delta^5\oplus \gamma^8\delta^6 =  (a\odot c) \wedge (b \odot c)$. This, among other results, invalidates the method given in Section 5.3 of~\cite{brunsch2014modeling} to detect unfeasible constraints.}
Other only necessary and only sufficient conditions for consistency were given, \eg in~\cite{komenda2011application,zorzenon2020bounded,vspavcek2021analysis}.

In \cite{lee2005extended}, Lee and Park introduced negative event graphs (NEGs), which are a slightly more general class of systems compared to P-TEGs, and formulated necessary and sufficient conditions for consistency in strongly connected NEGs.
These conditions were extended to arbitrary NEGs by Munier Kordon in \cite{munier2011graph}, where a weakly polynomial-time algorithm to check the consistency property was proposed.
\zor{The approach by Lee, Park, and Munier Kordon is different from the one by Declerck.
Indeed, \cite{lee2005extended,munier2011graph} proved that consistency is equivalent to the existence of consistent \emph{periodic} trajectories -- in which each transition $t_i$ fires every $\lambda_i$ time units -- and then focused on finding efficient algorithms to compute them.}
Since P-TEGs form a subclass of NEGs, the algorithm found in \cite{munier2011graph} solves the consistency verification problem in P-TEGs as well.
The latter result came to our attention only after the publication of \cite{zorzenon2024consistency}, \zor{which proposes a different algorithm to check consistency}.

\zor{The results from \cite{lee2005extended,munier2011graph} are however only applicable to the case of \emph{loose initial conditions}, according to which initial tokens are allowed to contribute to the firing of transitions at any time, independently from the time window associated to their initial places.
Although loose initial conditions are suitable for some applications, such as for the analysis of manufacturing systems in periodic regimes, in which the influence of the initial conditions is negligible, they may be overly permissive in others.
For example, they are inadequate for modeling the fact that some machines in a manufacturing system have already been processing a part for a time $\tau\geq 0$ before the initial time $t_0$.
In this case, \emph{strict initial conditions} must be used, which allow to specify the arrival time of initial tokens in places.
In P-TEGs under strict initial conditions, the characterization given in \cite{lee2005extended,munier2011graph} does not hold.
In fact, not all consistent P-TEGs under strict initial conditions admit a consistent periodic trajectory (this is formally shown in \Cref{ex:pteg_nonperiodic}).}

The consistency verification problem in P-TEGs can be reformulated as the problem of verifying the existence of solutions in particular systems of infinitely many precedence (or potential) constraints.
Precedence constraints are inequalities of the form 
\begin{equation}\label{eq:single_potential}
    y \geq c + x,
\end{equation}
where $x$ and $y$ are real variables and $c$ is a real constant.
It is convenient to visualize them using precedence graphs, which are weighted directed graphs with one node for each variable in the system and an arc from node $x$ to node $y$ with weight $c$ for each inequality of the form \eqref{eq:single_potential}. 

In \cite[Hilfssatz (2.2.1)]{gallai1958maximum}, Gallai observed the following useful characterization: a system of finitely many precedence constraints admits real solutions if and only if the corresponding precedence graph does not contain a \emph{positive-weight circuit} \zor{(see \cite[Section 6]{SCHRIJVER20051} for a historical excursus on this and other results related to the shorted path algorithm)}.
This property was extended to specific systems of infinitely many precedence constraints in \cite{zorzenon2024consistency}.

\subsection{Our contributions}

The present paper, which represents an extension of \cite{zorzenon2024consistency}, generalizes \zor{the characterization discovered by Gallai} to arbitrary systems of precedence constraints.
We prove in particular that \emph{any} (possibly infinite) system of precedence constraints admits solutions if and only if the corresponding precedence graph contains no \emph{$\infty$-weight paths}, \ie the supremal weight of all paths connecting any two nodes is less than $+\infty$ (see \Cref{th:gallai} in \Cref{se:algebra}). %
\zor{We would like to emphasize the generality of this result, which holds even for non-locally finite graphs -- in which some nodes have infinitely many incoming or outgoing arcs -- as well as for graphs consisting of infinitely many (strongly) connected subgraphs.}

We then proceed, in \Cref{se:N_periodic_graphs,se:ultimately_periodic_graphs}, to analyze two specific classes of infinite precedence graphs that find applications in P-TEGs: $\N$-periodic graphs and ultimately periodic graphs.
$\N$-periodic graphs are obtained, roughly speaking, by placing a finite graph with $n$ nodes at each point of the lattice of natural numbers $\N$ (see, for instance, \Cref{fi:infinite_precedence_graph}).
The dynamics of P-TEGs with loose initial conditions evolves according to precedence constraints described by $\N$-periodic graphs.
P-TEGs with strict initial conditions are instead modeled by ultimately periodic graphs, which are an extension of $\N$-periodic graphs consisting of three parts: a negative-periodic part (analogous to an $\N$-periodic part but for the lattice of negative integers), a transient part, and a positive-periodic part (an $\N$-periodic graph).
See \Cref{fi:ultimately_periodic_graph} for an example of ultimately periodic graph.

In \Cref{se:N_periodic_graphs}, we prove that the presence of $\infty$-weight paths in $\N$-periodic graphs can be detected in strongly polynomial time.
\zor{Central for our proof is the sequence of matrices introduced by Declerck in \cite{5628259}; indeed, the presence of $\infty$-weight paths is equivalent to the divergence of the sequence, and we show that it is possible to decide whether the sequence converges or not in time $O(n^5)$, in the worst case.}
In \Cref{se:ultimately_periodic_graphs}, we extend the latter result to ultimately periodic graphs.
The consequence, illustrated in \Cref{se:P_TEGs}, is that consistency of P-TEGs, \zor{under either loose or strict} initial conditions, consistency can be verified in time $O(n^5)$.

To conclude, in \Cref{se:conclusions} we comment on interesting connections between the models discussed in this paper and others, such as vector addition systems with states.

\subsubsection*{Notation}
We denote sets $\R\cup\{-\infty\}$, $\R\cup\{+\infty\}$, $\R\cup\{-\infty,+\infty\}$ respectively by $\Rmax$, $\Rmin$, and $\Rbar$.
The sets of negative, nonnegative, and positive integers are denoted, respectively, by $\Z_{<0}$, $\No$, and $\N$.
Given $a,b\in\Z$, with $b\geq a$, $\dint{a,b}$ indicates the discrete interval $\{a,a+1,a+2,\ldots,b\}$.

\section{Infinite precedence graphs and constraints}\label{se:algebra}

The max-plus algebra is a mathematical framework that allows to conveniently translate graph-theoretical algorithms for the longest-path problem, such as the Bellman-Ford and Floyd-Warshall algorithms, into algebraic expressions.
Typically, the considered graphs are finite, but in this section we show that extending some results to infinite graphs is possible.
By exploiting the max-plus framework, we then prove the connection between the existence of solutions in (infinite) precedence constraints and the presence of $\infty$-weight paths in precedence graphs.

\subsection{Basic algebraic tools}

Before introducing the max-plus algebra, we recall the definition of an idempotent semiring.
An \textit{idempotent semiring} $(\D,\oplus,\otimes)$ consists of a set $\D$ endowed with an operation $\oplus$ (called addition), which is commutative, associative, idempotent (i.e., $a\oplus a = a$), and has neutral element $\varepsilon$, and an operation $\otimes$ (called multiplication), which is associative, distributive over $\oplus$, has neutral element $e$, and $\forall a\in\D$, $a\otimes \varepsilon=\varepsilon\otimes a = \varepsilon$.
The partial order relation $\succeq$ is defined by: $a\succeq b\ \Leftrightarrow \ a\oplus b = a$.
Any idempotent semiring is closed under finite additions and multiplications; if it is also closed under infinite additions and if $\otimes$ distributes over infinite additions, then we say that it is \textit{complete}.
In this case, given any $a\in\D$, the operator $ ^+$ applied to $a$ is defined by $a^+ = \bigoplus_{i\in\N} a^i$, where $a^1=a$ and $a^{i+1} = a\otimes a^i$.
The \textit{Kleene star} of $a$, $a^* = a^+ \oplus e$, has the following property (see, \eg \cite{baccelli1992synchronization}):
\begin{align}
\label{eq:a_star_a_star}a^*\otimes a^* = (a^*)^* =  a^*.
\end{align}

\begin{remark}\label{re:inequality_oplus}
The following equivalence holds in any idempotent semiring: 
\[a\succeq b \mbox{ and } a \succeq c\ \Leftrightarrow\ a \succeq b \oplus c.\]
In complete idempotent semirings, this property extends to the case of infinitely many inequalities, \ie for all $\X\subseteq \D$,
\[
    \forall x\in \X,\quad a\succeq x\quad \Leftrightarrow\quad a\succeq \bigoplus_{x\in \X} x,
\]
where $\bigoplus_{x\in\X} x$ indicates the supremum of set $\X$ according to relation $\succeq$.
For a proof of this fact, we refer to~\cite[Remark 2.2 (a)]{singer2003some}.
\end{remark}

The max-plus algebra $(\Rbar,\oplus,\otimes)$, where the operations $\oplus$ and $\otimes$ are defined for all $a,b\in\Rbar$ by 
\[
    a\oplus b=\max\{a,b\},\quad a\otimes b=
    \begin{dcases}
    a+b & \mbox{if } a\neq -\infty \mbox{ and } b\neq -\infty, \\
    -\infty & \mbox{otherwise,}
    \end{dcases}
\]
is a complete idempotent semiring.
On the other hand, $(\Rmax,\oplus,\otimes)$ is an example of an idempotent semiring that is not complete.
The operations of the max-plus algebra can be extended to matrices of finite and infinite dimensions.
Let $I,J$ be arbitrary countable sets; a matrix $A$ is a function $A:I\times J\rightarrow \Rbar$, where $A(i,j)$, denoted by $A_{ij}$, is an entry of $A$.
The collection of such matrices is denoted by $\Rbar^{I\times J}$ or, when $I=\dint{1,m}$ and $J=\dint{1,n}$, with $m,n\in\N$, by $\Rbar^{m\times n}$.
The set $\Rbar^{I\times \{1\}}$ of column vectors is simply indicated by $\Rbar^I$ (or $\Rbar^n$ when $I=\dint{1,n}$).
Given $A,B\in\Rbar^{I_1\times I_2}$, $C\in\Rbar^{I_2\times I_3}$, for all $i\in I_1$, $j\in I_2$, $h\in I_3$, we set
\[
    (A\oplus B)_{ij} = A_{ij} \oplus B_{ij},\quad (A\otimes C)_{ih} = \bigoplus_{k\in I_2} A_{ik}\otimes C_{kh}.
\]
With these definitions, $(\Rbar^{n\times n},\oplus,\otimes)$ and $(\Rbar^{\N\times \N},\oplus,\otimes)$ form two complete idempotent semirings, see \cite[Section 1.4]{droste2009handbook}.
In such semirings, the neutral element for $\oplus$ is the matrix $\mathcal{E}$ whose entries are all $-\infty$, and the neutral element for $\otimes$ is the matrix $E$ such that $E_{ii}=0$ for all $i$ and $E_{ij} = -\infty$ for all $i\neq j$.
The scalar-matrix product in the max-plus algebra is defined for all scalars $\lambda\in\Rbar$ and matrices $A\in\Rbar^{I\times J}$ by: for all $i\in I$, $j\in J$,
\[
    (\lambda\otimes A)_{ij} = \lambda \otimes A_{ij}.
\]
According to the definition of $\succeq$ in idempotent semirings, given two matrices $A$ and $B$ of the same size, we have $A\succeq B$ if and only if, for all $i,j$, $A_{ij}\geq B_{ij}$; to simplify notation we will always write "$A\geq B$" in place of "$A \succeq B$".
When clear from the context, we will also omit the multiplication sign "$\otimes$". %

\subsection{Precedence constraints}

\emph{Precedence} (or \emph{potential}) \emph{constraints} are systems of inequalities of the form
\begin{equation}\label{eq:potential_inequalities_linear}
    \forall i,j\in I\quad x_i \geq A_{ij} + x_j,
\end{equation}
where $A\in\Rmax^{I\times I}$ is called \emph{difference bound matrix} and $x\in\R^I$ is a vector of real variables.
Precedence constraints can be written in the max-plus algebra as
\begin{equation}\label{eq:potential_inequalities}
    x \geq A\otimes x.
\end{equation}
The equivalence between~\eqref{eq:potential_inequalities_linear} and~\eqref{eq:potential_inequalities} is easily proven:
\[
        \eqref{eq:potential_inequalities_linear}  \Leftrightarrow \forall i,j\in I, \ x_i \geq A_{ij}\otimes x_j 
                                                 \ifftop{Remark~\ref{re:inequality_oplus}}  \forall i\in I,\ x_i\geq \bigoplus_{j\in I} A_{ij}\otimes x_j \ \Leftrightarrow \ \eqref{eq:potential_inequalities}.
\]

\begin{example}\label{ex:potential_inequalities_1}
Consider the following system of infinitely many inequalities in infinitely many variables $x_1,x_2,\ldots\in\R$ (written using standard notation):
    \begin{equation}\label{eq:example_potential_inequalities}
        \forall k\in\N,\quad
        \left\{
            \begin{array}{rcl}
                x_{2k} &\geq& 0 + x_{2k-1}\\
                x_{2k+1}& \geq& -4 k + x_{1}\\
                x_{2k+2}& \geq& 2 k + x_{2}\\
                x_{2k-1}& \geq& 3 + x_{2k+1}.
            \end{array} 
        \right.
    \end{equation}
The same inequalities can be expressed in the form~\eqref{eq:potential_inequalities}, by defining the difference bound matrix
    \[
\setlength{\arraycolsep}{0pt}
        A = 
        \begin{bNiceMatrix}[columns-width=auto]
            \cdot & \cdot &     3 & \cdot & \cdot & \cdot & \cdot & \cdot & \cdots \\
            0 & \cdot & \cdot & \cdot & \cdot & \cdot & \cdot & \cdot  & \cdots\\
            -4 & \cdot & \cdot & \cdot &     3 & \cdot & \cdot & \cdot  & \cdots\\
            \cdot &     2 &     0 & \cdot & \cdot & \cdot & \cdot & \cdot  & \cdots\\
            -8 & \cdot & \cdot & \cdot & \cdot & \cdot &     3 & \cdot  & \cdots\\
            \cdot &     4 & \cdot & \cdot &     0 & \cdot & \cdot & \cdot  & \cdots\\
            -12 & \cdot & \cdot & \cdot & \cdot & \cdot & \cdot & \cdot  & \cdots\\
            \cdot &     6 & \cdot & \cdot & \cdot & \cdot &     0 & \cdot & \cdots\\ 
            \vdots & \vdots & \vdots & \vdots & \vdots & \vdots & \vdots & \vdots & \ddots 
        \end{bNiceMatrix},
    \]
where each "$\cdot$" in the matrix stands for $-\infty$.
\end{example}

The following result will be useful later.

\begin{theorem}{\cite[Theorem 4.70]{baccelli1992synchronization}}\label{th:Kleene_expressions}
For all $A\in\Rbar^{I\times I}$ and $x\in\Rbar^{I}$, 
\[
    x \geq A\otimes x \ \Leftrightarrow\ x = A^* \otimes x.\qedlineend
\]
\end{theorem}

\zor{In the following sections, we will study the existence of real solutions of precedence constraints.
We remark that we are not interested in solutions with entries equal to $-\infty$ or $+\infty$, since they are not useful for our purposes.}

\subsection{Precedence graphs}

A (finite or infinite) \emph{directed graph} (or \emph{digraph}) is a pair $(\nodes,\arcs)$, where $\nodes$ is a countable set called the set of nodes and $\arcs\subseteq \nodes\times\nodes$ is the set of arcs.
A path $\rho$ of length $|\rho|=\ell\in\N$ in $(\nodes,\arcs)$ is a finite sequence of nodes $(\rho_1,\rho_2,\ldots,\rho_{\ell+1})$ such that, $\forall i\in\dint{1,\ell}$, $(\rho_i,\rho_{i+1})\in \arcs$.
An \emph{elementary} path is one in which no node appears more than once, \ie $\forall i,j\in\dint{1,|\rho|+1}$, $\rho_i=\rho_j \Rightarrow i=j$.
A path $\rho$ is called a \emph{circuit} if its initial and final nodes coincide, \ie if $\rho_1=\rho_{|\rho|+1}$.
A circuit $\rho$ is elementary if the path $(\rho_1,\rho_2,\ldots,\rho_{|\rho|})$ is elementary.

A digraph $(\nodes',\arcs')$ is a (\emph{proper}) \emph{subgraph} of digraph $(\nodes,\arcs)$ if $\nodes'\subseteq \nodes$ ($\nodes'\subset \nodes$) and $\arcs' = \{(i,j)\in\arcs\mid i,j\in\nodes'\}$.
A digraph is said to be \emph{strongly} (resp., \emph{fully}) \emph{connected} if there exists a path (resp., arc) from any node to any other node of the digraph.
A digraph $(\nodes,\arcs)$ is \emph{connected} if the associated undirected graph, obtained by adding to $(\nodes,\arcs)$ an arc $(i,j)$ for each arc $(j,i)\in\arcs$ such that $(i,j)\not\in\arcs$, is strongly connected.
A strongly connected (resp., fully connected, connected) subgraph of digraph $(\nodes,\arcs)$ is \emph{maximal} if it is not a proper subgraph of another strongly connected (resp., fully connected, connected) subgraph of $(\nodes,\arcs)$.

A \emph{weighted directed graph} is a 3-tuple $(\nodes,\arcs,w)$, where $(\nodes,\arcs)$ is a directed graph, and $w:\arcs\rightarrow \R$ is a function that associates a weight $w((j,i))$ to every arc $(j,i)$ of the graph.
The weight $|\rho|_W\in\R$ of a path $\rho$ in a weighted directed graph is the (standard) sum of the weight of its edges, \ie
\[
    |\rho|_W = \sum_{i=1}^{|\rho|} w((\rho_i,\rho_{i+1})).
\]
If there exists an infinite sequence of paths $\rho_1,\rho_2,\dots$ in $(\nodes,\arcs,w)$ such that $\lim_{h\rightarrow+\infty} |\rho_h|_{W} = +\infty$, then we say that $(\nodes,\arcs,w)$ contains an \emph{$\infty$-weight path}.\footnote{Note the slight abuse of terminology: $\infty$-weight paths are not paths. Note also that any "path" with infinite weight would have infinite length, going against our definition of path.} 

Given a weighted directed graph $G = (\nodes,\arcs,w)$, we say that $A\in\Rmax^{I\times I}$ is the \emph{adjacency matrix} of $G$, and that $G$ is the \emph{precedence graph} of $A$, if $I = \nodes$ and there is an arc in $G$ from node $i$ to node $j$ of weight $w((i,j)) = A_{ji}$ if and only if $A_{ji} \neq -\infty$.
In this case, we write $G = \graph(A)$.
With these definitions, element $(j,i)$ of matrix $A^{\ell}$, respectively, $A^+$, corresponds to the supremal weight of all paths in $\graph(A)$ from node $i$ to node $j$ of length $\ell$, respectively, of any length.
Note that $\graph(A)$ contains an $\infty$-weight path from node $i$ to node $j$ if and only if $(A^+)_{ji} = (A^*)_{ji} = +\infty$; if no $\infty$-weight paths are present in $\graph(A)$, then all elements of $A^+$ and $A^* = A^+\oplus E$ belong to $\Rmax$.
We indicate by $\nonegset$ the set of all precedence graphs without $\infty$-weight paths:
\[
    \nonegset = \{\graph(A)\mid A\in\Rmax^{I\times I},\ I\mbox{ is a countable set},\ A^+\in\Rmax^{I\times I}\}.
\]

\begin{remark}\label{re:elementary_paths_circuits}
    Recall that, in finite graphs, there exists an $\infty$-weight path if and only if there exists an elementary circuit with positive weight.
    Therefore, for all $A\in\Rmax^{n\times n}$, $\graph(A)\in\nonegset$ if and only if $(A^+)_{ii}\leq 0$ and $(A^*)_{ii}=0$ for all $i\in\dint{1,n}$.

\zor{In infinite graphs, the existence of positive-weight circuits is not necessary to have $\infty$-weight paths. 
The presence of an $\infty$-weight path is in fact a necessary and sufficient condition for the existence of at least one of the following objects:
    \begin{itemize}
        \item an elementary circuit with positive weight,
        \item an infinite sequence of \emph{elementary} paths with infinite limit weight.
    \end{itemize}
    To show this, consider an infinite graph with an $\infty$-weight path but no elementary circuit with positive weight.
        By definition of $\infty$-weight paths, there is a sequence of paths $\rho_1,\rho_2,\dots$ with infinite limit weight.
    If the sequence contains a non-elementary path $\rho_k$, we can decompose it into a concatenation of elementary paths and elementary circuits.
    Then, because all circuits have non-positive weight, they can be eliminated from $\rho_k$ obtaining an elementary path $\rho_k'$ with larger or equal weight.}
\end{remark}

\begin{example}
    \zor{An example of infinite precedence graph with $\infty$-weight paths but without circuits (and, thus, without elementary circuits with positive weight) is shown in \Cref{fi:infinite_precedence_graph}.
    The sequence of elementary paths $\rho_1,\rho_2,\dots$ from node $1$ to node $2$ where $\rho_k=(1,\, 3,\, 5,\, \dots ,\, 2k-1,\, 2k, \, 2k-2, \, 2k-4,\, \dots,\, 2)$ has infinite limit weight, since $\lim_{k\rightarrow\infty}|\rho_k|_W =\lim_{k\rightarrow\infty} k-1 = +\infty$.
    Let $A\in\Rmax^{\N\times\N}$ be the incidence matrix of the graph in \Cref{fi:infinite_precedence_graph}.
    Since no circuit exists, we have $(A^+)_{ii} = -\infty$ for all $i\in\N$.
    On the other hand, $(A^+)_{21} = +\infty$.
    }
\begin{figure}[t]
    \centering
    \begin{tikzpicture}[node distance=2cm and 2cm,>=stealth',bend angle=45,double distance=.5mm,arc/.style={->,>=stealth'},place/.append style={minimum size=.3cm}]

\foreach \z in {1,2,3,4,5}
{
\pgfmathtruncatemacro{\za}{\z*2-1}
\pgfmathtruncatemacro{\zb}{\z*2}
\node [place] (ntop\z) at (1.5*\z,0) {$\za$};
\node [place] (nbot\z) at (1.5*\z,-1*1.5) {$\zb$};
\draw [arc] (ntop\z) to node[auto] {$0$} (nbot\z);
}
\draw [arc] (ntop1) to node[auto] {$2$} (ntop2);
\draw [arc] (ntop2) to node[auto] {$2$} (ntop3);
\draw [arc] (ntop3) to node[auto] {$2$} (ntop4);
\draw [arc] (ntop4) to node[auto] {} (ntop5);

\draw [arc] (nbot2) to node[auto] {$-1$} (nbot1);
\draw [arc] (nbot3) to node[auto] {$-1$} (nbot2);
\draw [arc] (nbot4) to node[auto] {$-1$} (nbot3);
\draw [arc] (nbot5) to node[auto] {} (nbot4);

\fill [white] (1.5*4.5,1.5*.5) rectangle (1.5*5.5,1.5*-1.5);
\node (dots2) at (1.5*5,1.5*-.5) {\textbf{\dots}};

\end{tikzpicture}
    \caption{Infinite precedence graph with an $\infty$-weight path but without circuits with positive weight.}\label{fi:infinite_precedence_graph}
\end{figure}
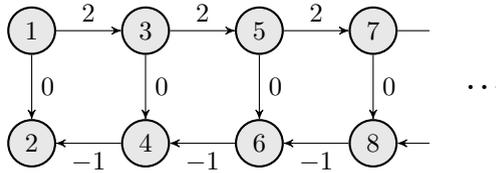
\end{example}

\subsection{Extension of Gallai's observation to infinite precedence graphs}

The following core result provides a necessary and sufficient condition for the existence of real solutions of precedence constraints, based on the precedence graph of the related difference bound matrix.
It was discovered by Gallai for finite difference bound matrices~\cite[Hilfssatz (2.2.1)]{gallai1958maximum}.
To the best of our knowledge, its extension to the case of generic infinite matrices is presented here for the first time. %
Its proof also provides a procedure for obtaining a solution when one exists.

\begin{theorem}\label{th:gallai}
    For a given $A\in \Rmax^{I\times I}$, inequality $x\geq A\otimes x$ admits a solution $x\in\R^I$ if and only if $\graph(A)\in\nonegset$.\footnote{This theorem remains valid if $\R$ is substituted by $\Z$, but not by $\Q$, because $\Q\cup\{\pm\infty\}$ is not closed for infinitely many $\oplus$.}
\end{theorem}

\zor{In order to prove the theorem, we need some technical definitions and lemmas.

\begin{definition}
Let $\preceq$ be a partial order relation on a set $I$.
We say that $\preceq$ is a \emph{good order} if the following conditions are satisfied:
\begin{itemize}
    \item $\preceq$ is a \emph{well-order} on $I$, i.e., a total order relation\footnote{A \emph{total order relation} $\preceq$ is a binary relation that is reflexive ($a\preceq a$), antisymmetric $(a\preceq b \wedge b\preceq a\Rightarrow a=b$), and transitive ($a\preceq b\wedge b\preceq c\Rightarrow a\preceq c$) such that, for every element $a,b$, either $a\preceq b$ or $b\preceq a$.} for which every nonempty subset of $I$ has a least element,
    \item for any element $a\in I$, there are finitely many elements less than or equal to $a$, i.e., $|\{x\in I\mid x\preceq a\}|\in\N$.\qedhere
\end{itemize}
\end{definition}
An example of well-order on $\N$ that is not good is the relation $\preceq$, defined such that $1\preceq 3\preceq 5\preceq 7\preceq \ldots \preceq 2\preceq 4\preceq 6\preceq 8\preceq \ldots$
\begin{lemma}\label{le:gallaiaux1}
    For any countable set $I$, there exists a good order on $I$.
\end{lemma}
\begin{proof}
By definition of countable set, there exists an injective function $f:I\rightarrow\N$.
Define the order $\preceq$ in $I$ such that, for all $a,b\in I$, $a\preceq b$ if and only if $f(a)\leq f(b)$.
Since $\leq$ is a good order on $\N$, then clearly the same can be said for the relation $\preceq$ on $I$.
\end{proof}
}

\zor{An example of good order on $\N\times\N$ is the one derived by the well-known Cantor pairing function $f((m,n)) = \frac{1}{2}(m+n-2)(m+n-1)+m$, which is a bijection from $\N\times\N$ to $\N$ \cite[page 169]{hopcroft1979introduction}.}

\begin{definition}
\zor{Let $A\in\Rbar^{I\times I}$, and let $\preceq$ be a good order on the set of pairs $I\times I$ (whose existence is guaranteed by \Cref{le:gallaiaux1}, since $I\times I$ is a countable set).
We define $\Phi:\Rbar^{I\times I}\rightarrow \Rbar^{I\times I}$ to be the function that, when applied on matrix $A$, returns $A^*$ if there is no pair $(i,j)$ such that $(A^*)_{ij}=-\infty$ and $(A^*)_{ji} \neq -\infty$,\footnote{In this case, if $A\in\Rmax^{I\times I}$, then $\graph(A)$ consists only of disjoint (i.e., not connected by arcs) maximal strongly connected subgraphs.}}
and otherwise returns
\[
    \Phi(A) = A^* \oplus (-(A^*)_{ji}) \otimes U_{(i,j)},
\]
\zor{where $(i,j)$ is the least pair, according to $\preceq$}, such that $(A^*)_{ij} = -\infty$ and $(A^*)_{ji}\neq-\infty$, $-(A^*)_{ji}$ simply indicates the opposite of $(A^*)_{ji}$ in the standard sense (\ie if $(A^*)_{ji}\in\R$, then $(A^*)_{ji}\otimes (-(A^*)_{ji})=0$, and if $(A^*)_{ji}=+\infty$, then $-(A^*)_{ji} = -\infty$), and $U_{(i,j)}\in\{-\infty,0\}^{I\times I}$ is $0$ in position $(i,j)$ and $-\infty$ everywhere else.
\end{definition}
Note that the exact definition of $\Phi$ depends on the \zor{chosen good order}, but the following discussion holds true for any admissible choice.

\zor{
\begin{figure}[t]
\centering
\subfloat[$\graph(A)$.]{\label{fi:finite_precedence_graph}\resizebox{.7\textwidth}{!}{\begin{tikzpicture}[node distance=1cm and 2cm,>=stealth',bend angle=45,double distance=.5mm,arc/.style={->,>=stealth'},place/.append style={minimum size=.3cm}]

\node [place] (1) at (0,0) {$1$};
\node [place,right=of 1] (2) {$2$};
\node [place,right=of 2] (3) {$3$};
\node [place,right=of 3] (4) {$4$};
\node [place,right=of 4] (5) {$5$};

\draw [arc] (1) to (2);
\draw [arc] (2) to (3);
\draw [arc] (3) to [bend left=30] node [auto] {$-3$} (1);
\draw [arc] (3) to node [auto] {$-1$} (4);
\draw [arc] (5) to node [auto] {$2$} (4);

\end{tikzpicture}}}
\hfill
\subfloat[$\graph(A^*)$.]{\label{fi:finite_precedence_graph_star}\resizebox{.7\textwidth}{!}{\begin{tikzpicture}[node distance=1cm and 2cm,>=stealth',bend angle=45,double distance=.5mm,arc/.style={->,>=stealth'},place/.append style={minimum size=.3cm}]

\node [place] (1) at (0,0) {$1$};
\node [place,right=of 1] (2) {$2$};
\node [place,right=of 2] (3) {$3$};
\node [place,right=of 3] (4) {$4$};
\node [place,right=of 4] (5) {$5$};

\draw [arc] (1) to (2);
\draw [arc] (2) to (3);
\draw [arc] (3) to [bend left=30] node [auto] {$-3$} (1);
\draw [arc] (3) to node [auto] {$-1$} (4);
\draw [arc] (5) to node [auto] {$2$} (4);

\begin{scope}[very thick,myblue]
\draw [arc] (1) to [bend left=40] (3);
\draw [arc] (2) to [bend right=20] node [above] {$-3$} (1);
\draw [arc] (3) to [bend right=20] node [above] {$-3$} (2);
\draw [arc] (1) to [bend left=50] node [below] {$-1$} (4);
\draw [arc] (2) to [bend left=40] node [auto] {$-1$} (4);

\foreach \x in {1,...,5}{
\draw [arc] (\x) to [out=-90+15,in=-90-15,loop] (\x);
}
\end{scope}

\end{tikzpicture}}}
\hfill
\subfloat[$\graph(\Phi(A))$.]{\label{fi:finite_precedence_graph_phi}
    \resizebox{.7\textwidth}{!}{\begin{tikzpicture}[node distance=1cm and 2cm,>=stealth',bend angle=45,double distance=.5mm,arc/.style={->,>=stealth'},place/.append style={minimum size=.3cm}]

\node [place] (1) at (0,0) {$1$};
\node [place,right=of 1] (2) {$2$};
\node [place,right=of 2] (3) {$3$};
\node [place,right=of 3] (4) {$4$};
\node [place,right=of 4] (5) {$5$};

\draw [arc] (1) to (2);
\draw [arc] (2) to (3);
\draw [arc] (3) to [bend left=30] node [auto] {$-3$} (1);
\draw [arc] (3) to node [auto] {$-1$} (4);
\draw [arc] (5) to node [auto] {$2$} (4);

\draw [arc] (1) to [bend left=40] (3);
\draw [arc] (2) to [bend right=20] node [above] {$-3$} (1);
\draw [arc] (3) to [bend right=20] node [above] {$-3$} (2);
\draw [arc] (1) to [bend left=50] node [below] {$-1$} (4);
\draw [arc] (2) to [bend left=40] node [auto] {$-1$} (4);

\foreach \x in {1,...,5}{
\draw [arc] (\x) to [out=-90+15,in=-90-15,loop] (\x);
}

\begin{scope}[very thick,myblue]
\draw [arc] (4) to [bend left=50] node [above] {$1$} (1);
\end{scope}

\end{tikzpicture}}
}
\caption{Application of function $\Phi$ illustrated on a finite precedence graph. Whenever not indicated, the weight of arcs is $0$. \myblue{Blue} thick arcs highlight the difference between $\graph(A)$ and $\graph(A^*)$, and between $\graph(A^*)$ and $\graph(\Phi(A))$.}
\label{fi:phi}
\end{figure}
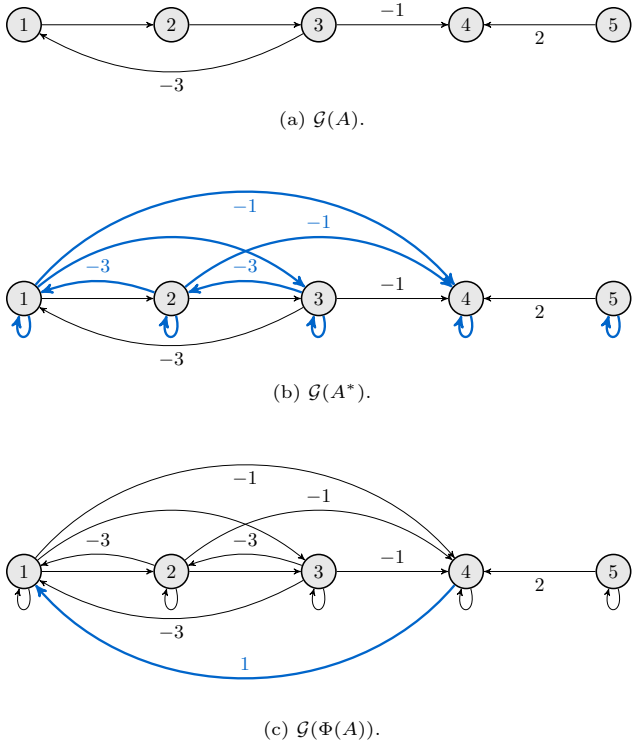
\begin{example}\label{ex:Phi}
    Consider matrix
    \[
        A = 
        \begin{bmatrix}
            -\infty&  -\infty&    -3&  -\infty&  -\infty\\
            0&  -\infty&  -\infty&  -\infty&  -\infty\\
            -\infty&     0&  -\infty&  -\infty&  -\infty\\
            -\infty&  -\infty&    -1&  -\infty&     2\\
            -\infty&  -\infty&  -\infty&  -\infty&  -\infty\\
        \end{bmatrix}.
    \]
    The precedence graphs $\graph(A)$ and $\graph(A^*)$ are illustrated in \Cref{fi:finite_precedence_graph} and \Cref{fi:finite_precedence_graph_star}, respectively.
    The set of pairs of nodes $(i,j)$ such that $(A^*)_{ij}=-\infty$ and $(A^*)_{ji}\neq -\infty$ is $S = \{(1,4),\,(2,4),\,(3,4),\,(5,4)\}$.
    Let the good order $\preceq$ on $\dint{1,5}\times\dint{1,5}$ be defined such that $(i,j)\preceq (h,k)$ if and only if $f((i,j))\leq f((h,k))$, where $f$ is the Cantor pairing function restricted to domain $\dint{1,5}\times\dint{1,5}$.
    By direct computation, we can observe that the minimum element of $S$ according to $\preceq$ is $(1,4)$.
    Then, if $\Phi$ is defined based on $\preceq$, $\Phi(A)$ is identical to $A^*$ except for element $(1,4)$, where $(A^*)_{14} = -\infty$ and $(\Phi(A))_{14} = 1$.
    The precedence graph $\graph(\Phi(A))$ is shown in \Cref{fi:finite_precedence_graph_phi}.
\end{example}
}

\zor{Observe that,} for all connected but not strongly connected $\graph(A)\in\nonegset$, compared to $\graph(A^*)$, $\graph(\Phi(A))$ has an additional arc $(j,i)$, where $i$ and $j$ belong to different maximal fully connected subgraphs of $\graph(A^*)$, as illustrated in \Cref{fi:cloud_graph}.
Because of the presence of circuit $(i,j,i)$ (with zero weight), $\graph(\Phi(A))$ has in this case at least one maximal strongly connected subgraph less than $\graph(A^*)$.
\begin{figure}[h]
    \centering
    \begin{tikzpicture}
\node (A) [cloud,draw,fill=LightGray,cloud puffs=10,cloud puff arc=120, aspect=2] {$\graph_1$}; 
\node (B) [cloud,draw,fill=LightGray,cloud puffs=15,cloud puff arc=120, aspect=2,right=of A] {$\graph_2$}; 
\node (C) [cloud,draw,fill=LightGray,cloud puffs=7,cloud puff arc=120, aspect=2,right=of B] {$\graph_3$}; 
\node (g) [left=of A] {$\graph(A^*):$}; 
{
\small
\draw [arc] ($(A)+(.5,0)$) to [bend left=10] ($(B)+(-.5,0)$); 
\draw [arc] ($(A)+(.5,.3)$) to [bend left]  ($(B)+(-.5,.2)$); 
\draw [arc] ($(A)+(.5,-.3)$) to [bend right]  ($(B)+(-.5,-.2)$); 
\draw [arc] ($(B)+(.5,0)$) to [bend left=40] node [auto] {$(A^*)_{ji}$} ($(C)+(-.5,0)$); 
}

\node (A) [cloud,draw,fill=LightGray,cloud puffs=10,cloud puff arc=120, aspect=2,below=of A] {$\graph_1$}; 
\node (B) [cloud,draw,fill=LightGray,cloud puffs=15,cloud puff arc=120, aspect=2,right=of A] {$\graph_2$}; 
\node (C) [cloud,draw,fill=LightGray,cloud puffs=7,cloud puff arc=120, aspect=2,right=of B] {$\graph_3$}; 
\node (g) [left=of A] {$\graph(\Phi(A)):$}; 
{
\small
\draw [arc] ($(A)+(.5,0)$) to [bend left=10] ($(B)+(-.5,0)$); 
\draw [arc] ($(A)+(.5,.3)$) to [bend left]  ($(B)+(-.5,.2)$); 
\draw [arc] ($(A)+(.5,-.3)$) to [bend right]  ($(B)+(-.5,-.2)$); 
\draw [arc] ($(B)+(.5,0)$) to [bend left=40] node [auto] {$(A^*)_{ji}$} ($(C)+(-.5,0)$); 
\draw [arc] ($(C)+(-.5,0)$) to [bend left=40] node [auto] {$-(A^*)_{ji}$} ($(B)+(.5,0)$); 
}

\end{tikzpicture}
    \caption{Schematic representation of $\graph(A^*)$ and $\graph(\Phi(A))$ for an example of graph $\graph(A^*)$ consisting of three maximal strongly connected subgraphs $\graph_1,\graph_2,\graph_3$. Compared to $\graph(A^*)$, $\graph(\Phi(A))$ has an additional arc $(j,i)$ of weight $-(A^*)_{ji}$, where $i$ and $j$ belong to $\graph_2$ and $\graph_3$, respectively. Observe that $\graph(\Phi(A))$ has only two maximal strongly connected subgraphs.}\label{fi:cloud_graph}
\end{figure}
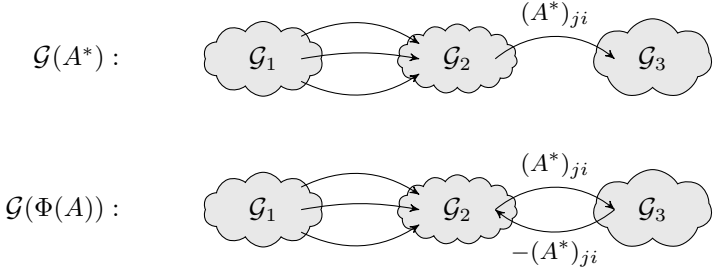

\begin{lemma}\label{le:gallaiaux3}
\zor{For all $A\in\Rmax^{I\times I}$ such that $\graph(A)\in\nonegset$, $\graph(\Phi(A))\in\nonegset$.}
\end{lemma}
\begin{proof}
\zor{We will prove the lemma by assuming that $\graph(A)$ is connected; the reasoning can easily be extended to the case of nonconnected graph $\graph(A)$.}
Suppose that $\graph(A)\in\nonegset$ and assume, by means of contradiction, that $\graph(\Phi(A))\not\in \nonegset$.
If $\graph(A)$ is strongly connected, we immediately have a contradiction, since $\Phi(A) = A^*\in\Rmax^{I\times I}$.
If $\graph(A)$ is not strongly connected, then $\graph(\Phi(A))$ is identical to $\graph(A^*)$ except for having an additional arc $(j,i)$.
Since, from the hypothesis and~\eqref{eq:a_star_a_star}, $\graph(A^*)$ does not contain $\infty$-weight paths, the $\infty$-weight paths in $\graph(\Phi(A))$ must use the arc $(j,i)$, \ie all sequences of paths $\rho_1,\rho_2,\dots$ with $\lim_{k\rightarrow \infty} |\rho_k|_W = +\infty$ must satisfy: $\exists h\in\N$ such that $\forall k\geq h$, $\rho_k$ contains the arc $(j,i)$.
Recalling \Cref{re:elementary_paths_circuits}, we need to examine the existence of two types of objects: i. elementary circuits with positive weight, and ii. sequences of elementary paths with infinite limit weight.

i. Suppose that there is an elementary circuit with positive weight formed by concatenating arc $(j,i)$ with an elementary path $\rho$ from $i$ to $j$ consisting of arcs from $\graph(A^*)$.
Then, the weight of this path satisfies
\[
    |\rho|_{W} \otimes |(j,i)|_W \leq (A^*)_{ji} \otimes (-(A^*)_{ji}) = 0.
\]

ii. Assume that there is an infinite sequence of elementary paths $\rho_1,\rho_2,\dots$ from $\ell$ to $m$ passing through arc $(j,i)$ (exactly once).
Then, the weight of each path $\rho_k$ is bounded from above by
\[
    (A^*)_{mi} \otimes |(j,i)|_{W} \otimes (A^*)_{j\ell} = (A^*)_{mi} \otimes (-(A^*)_{ji}) \otimes (A^*)_{j\ell} \in \R.
\]

Thus, neither elementary circuits with positive weight nor sequences of paths with infinite limit weight exist, which implies a contradiction and, in turn, that $\graph(\Phi(A))\in\nonegset$.
\end{proof}

\begin{definition}
For all $A\in\Rbar^{I\times I}$, we define $A^\circledast\in\Rbar^{I\times I}$ as
\[
    A^\circledast = \lim_{k\rightarrow \infty} \Phi^k(A),
\]
where $\Phi^0(A) = A$ and $\Phi^{k+1}(A) = \Phi(\Phi^k(A))$.
\end{definition}

\zor{
\begin{example}\label{ex:Phicont}
    The repeated application of function $\Phi$ on matrix $A$ from \Cref{ex:Phi} results in
    \[
        \Phi^1(A) = 
\begin{bmatrix}
    0    &-3    &-3     &1  &-\infty\\
    0     &0    &-3  &-\infty  &-\infty\\
    0     &0     &0  &-\infty  &-\infty\\
    -1    &-1    &-1     &0     &2\\
    -\infty  &-\infty  &-\infty  &-\infty&     0
\end{bmatrix},\quad
        \Phi^2(A) = 
\begin{bmatrix}
    0     &0     &0     &1     &3\\
    0     &0     &0     &1     &3\\
    0     &0     &0     &1     &3\\
    -1    &-1    &-1     &0     &2\\
    -3  &-\infty  &-\infty  &-\infty     &0
\end{bmatrix},
    \]
    and, for all $k\geq 3$,
    \[
        \Phi^k(A) = 
        (\Phi^2(A))^* = 
\begin{bmatrix}
    0     &0     &0     &1     &3\\
    0     &0     &0     &1     &3\\
    0     &0     &0     &1     &3\\
    -1    &-1    &-1     &0     &2\\
    -3    &-3    &-3    &-2     &0
\end{bmatrix}.
    \]
    Thus, $A^\circledast = \Phi^3(A)$.
\end{example}
}

\zor{
    \begin{lemma}\label{le:gallaiaux2}
    For all $A\in\Rmax^{I\times I}$ such that $\graph(A)\in\nonegset$:
    \begin{itemize}
        \item $A^\circledast\geq A$,
        \item if $i,j\in I$ belong to the same maximal connected subgraph of $\graph(A)$, then $(A^\circledast)_{ij}\neq-\infty$,
        \item $A^\circledast\in\Rmax^{I\times I}$ and $\graph(A^\circledast)\in\nonegset$.\qedhere
    \end{itemize}
\end{lemma}}
\begin{proof}
Clearly, $A^\circledast\geq A$ because $\Phi(A)\geq A^*\geq A$.

\zor{In the rest of the proof, we will assume, without loss of generality, that $\graph(A)$ is connected.
If $\graph(A)$ is not connected, the following reasoning can be repeated for each maximal connected subgraphs of $\graph(A)$.
Under this assumption, the second statement of the lemma becomes "for all $i,j\in I$, $(A^\circledast)_{ij}\neq -\infty$", and the third one "$A^\circledast\in\R^{I\times I}$ and $\graph(A^\circledast)\in\nonegset$".}

\zor{Since $A^\circledast\geq \Phi^k(A)$ for all $k\in\No$, to prove the second statement of the lemma it will be sufficient to show that, for any pair of nodes $i,j\in I$, there exists a $k\in\No$ such that $(\Phi^k(A))_{ji}\neq -\infty$.
Since $\graph(A)$ is connected, there exists a finite sequence of pairs $(k_1,k_2),\,(k_2,k_3),\,\ldots,\,(k_{h-1},k_h)$ such that $k_1=i$, $k_h=j$ and, for all $\ell\in\dint{1,h-1}$, either $(k_\ell,k_{\ell+1})$ or $(k_{\ell+1},k_{\ell})$ is an arc of $\graph(A)$.
Let $f:I\times I\rightarrow \N$ be the injection defining the good order $\preceq$ used to compute $\Phi(A)$, and let
\[
    m = \max\{f((k_{\ell+1},k_{\ell})) \mid \ell\in\dint{1,h-1}\}\in\N.
\]
Because the $k$-th application of $\Phi$ adds an arc $(k_{\ell},k_{\ell+1})$ to $\graph(\Phi^{k-1}(A)^*)$ if $(k_{\ell+1},k_{\ell})$ is an arc of $\graph(\Phi^{k-1}(A)^*)$, $(k_{\ell},k_{\ell+1})$ is not, and $(k_{\ell+1},k_{\ell})$ is minimal according to $\preceq$, in $k\leq m$ applications of $\Phi$ there will be a path from $i$ to $j$ in $\graph(\Phi^k(A))$.
Therefore, the precedence graph of $\Phi^{k+1}(A)$ will contain an arc from $i$ to $j$, since $\Phi^{k+1}(A)\geq (\Phi^{k}(A))^*$.
}

It remains to be proven that \zor{$A^\circledast\in\R^{I\times I}$ and} $\graph(A^\circledast)\in\nonegset$.
\zor{By using \Cref{le:gallaiaux3}, it is easy to prove by induction that $\graph(\Phi^k(A))\in\nonegset$ for all $k\in\No$.
We have already proven that $A^\circledast\in(\R\cup\{+\infty\})^{I\times I}$.
Observe that it is not possible that $A^\circledast\in\R^{I\times I}$ and $\graph(A^\circledast)\not\in\nonegset$, since $(A^\circledast)^*=A^\circledast$.
Therefore, either $A^\circledast\in\R^{I\times I}$ and $\graph(A^\circledast)\in\nonegset$ or $(A^\circledast)_{ij} = +\infty$ for some $i,j\in I$.
If $(A^\circledast)_{ij}=+\infty$ for some $i,j$, then $(A^\circledast)_{ii}=+\infty$, because $(A^\circledast)_{ii} = ((A^\circledast)^*)_{ii}\geq (A^\circledast)_{ji} \otimes (A^\circledast)_{ij}$ and $(A^\circledast)_{ij}\in\R\cup\{+\infty\}$.
However, remember from \Cref{le:gallaiaux3} that, for all $k\in\No$,} all circuits in $\graph(\Phi^k(A))$ have weight at most $0$.
Since the limit of sequence $(\Phi^0(A)^*)_{ii},\,(\Phi^1(A)^*)_{ii},\dots\in[-\infty,0]$ cannot exceed $0$ for all $i$, we conclude that $\graph(A^\circledast)$ cannot contain circuits with positive weight, and, as a consequence, that \zor{$A^\circledast\in\R^{I\times I}$ and} $\graph(A^\circledast)\in\nonegset$.
\end{proof}

We are now ready to prove \Cref{th:gallai}.

\begin{proof}[Proof of \Cref{th:gallai}]
The implication "$\Rightarrow$" has an elementary proof by contrapositive.
Suppose that $\graph(A)\not\in \nonegset$; then, $(A^*)_{ij}=+\infty$ for some $i,j\in I$.
Recall from~\Cref{th:Kleene_expressions} that $x\geq A x$ if and only if $x = A^*  x$. 
The $i$-th equation of $x = A^* x$ reads
\[
        x_i = \bigoplus_{k\in I} (A^*)_{ik} x_k 
        = (A^*)_{ij} x_j \oplus \bigoplus_{k\in I\setminus\{j\}} (A^*)_{ik} x_k.
\]
If $x_j\in\R$, then $x_i \geq (A^*)_{ij} x_j = +\infty$, implying that no real solution of $x\geq A x$ exists.

We prove the direction "$\Leftarrow$" in a constructive way.
Take $x$ as a max-plus linear combination of columns of matrix $A^*$:
\[
    x = \bigoplus_{j\in I} \alpha_j (A^*)_{\cdot j},
\]
where $\alpha_j\in\Rmax$ for all $j\in I$.
Such an $x$ always forms a solution of $x = A^* x$ (even when $\graph(A)\not\in\nonegset$) because, from the distributivity of $\otimes$ over $\oplus$, the commutativity of the scalar-matrix $\otimes$, and property $a^*\otimes a^* = a^*$,
\[
    A^* x = A^* \otimes \left( \bigoplus_{j\in I} \alpha_j (A^*)_{\cdot j} \right)
          = \bigoplus_{j\in I} \alpha_j A^* (A^*)_{\cdot j} 
          = \bigoplus_{j\in I} \alpha_j (A^*)_{\cdot j} = x.
\]

The more technical part of the proof is to find $\alpha_j$, $j\in I$, such that $x$ is \zor{a vector with only real (i.e., finite)} elements when $\graph(A)\in\nonegset$.

\zor{In case matrix $A\in\Rmax^{n\times n}$ has finitely many elements}, then taking any $\alpha_1,\dots,\alpha_n\in\R$ works since, for all $i\in\dint{1,n}$,
\[
    x_i = \bigoplus_{j\in\dint{1,n}} \alpha_j (A^*)_{ij} \geq \alpha_i (A^*)_{ii} \eqtop{\Cref{re:elementary_paths_circuits}} \alpha_i\otimes 0 = \alpha_i\in\R,
\]
\zor{which implies that $x_i>-\infty$, and} $x_i\neq +\infty$ as $\R$ is closed for finitely many $\oplus$'s.
\zor{However,} when $I$ is an infinite set, taking arbitrary real $\alpha_j$ does not always provide a real solution $x$ because $\R$ is not closed for infinite $\oplus$'s.
We thus need to "dampen" the growth of elements in columns of $A^*$ to avoid diverging to $+\infty$.

Our strategy to achieve this is based on the following observation:
if $\graph(A)\in\nonegset$ is strongly connected, then for all $j\in I$, $(A^*)_{\cdot j}\in \R^I$.
Indeed, if $\graph(A)$ is strongly connected, then $\graph(A^*)$ is fully connected, \ie $(A^*)_{ij}\neq -\infty$ for all $i,j\in I$.
Therefore, in the case of a strongly connected precedence graph, a real solution $x$ is given by picking an arbitrary $\overline{j}\in I$ and choosing
\begin{equation}\label{eq:alpha}
    \alpha_j = 
    \begin{dcases}
        0 & \mbox{if }j = \overline{j},\\
        -\infty & \mbox{else.}
    \end{dcases}
\end{equation}
\zor{This choice corresponds to the solution $x = (A^*)_{\cdot\overline{j}}\in\R^I$.}

This technique is easily generalized to the case when $\graph(A)$ consists of the union of disjoint (\ie not connected by any arc) maximal strongly connected subgraphs $\graph_{1},\graph_{2},\dots$
Here, element $(i,j)$ of $A^*$ is real if and only if $i$ and $j$ belong to the same maximal strongly connected subgraph.
Then, to generate a real solution $x$ it is sufficient to select arbitrarily one node $\overline{j}_k$ from each subgraph $\graph_k$ and define
\begin{equation}\label{eq:alpha1}
    \alpha_j =
    \begin{dcases}
        0 & \mbox{if } j \in \{ \overline{j}_1,\,\overline{j}_2,\ldots\},\\
        -\infty & \mbox{else.}
    \end{dcases}
\end{equation}
\zor{In this way, we get the solution $x = \bigoplus_{k=1,2,\ldots} (A^*)_{\cdot \overline{j}_k}\in\R^I$.}

We now consider the case when $\graph(A)$ is connected but not strongly connected.
\zor{As seen in \Cref{le:gallaiaux2}, matrix $A^\circledast$ satisfies $A^\circledast\geq A$, $A^\circledast\in \R^{I\times I}$, and $\graph(A^\circledast)\in\nonegset$.}
\zor{Therefore,} any real solution of $x\geq A^\circledast x$ (which can be obtained \zor{by defining $\alpha_j$ as in \eqref{eq:alpha}}, as any fully connected graph is strongly connected) also solves $x\geq A x$, since $x\geq A^\circledast x = (A^\circledast\oplus A)x = A^\circledast x \oplus Ax\geq Ax$.

If $\graph(A)$ is not even connected, then %
the graph $\graph(A^\circledast)$ will not be fully connected, but will consist of the union of disjoint fully connected subgraphs.
\zor{Therefore, we can use the same strategy defined in \eqref{eq:alpha1} to get a} $x\in\R^I$ that solves $x\geq A^\circledast x\geq A x$.
\end{proof}

\begin{example}\label{ex:potential_inequalities_2}
    \Cref{fi:graph_to_circledast} illustrates the precedence graph of matrix $A$ defined in \Cref{ex:potential_inequalities_1}.
    By inspecting the graph, one can easily be convinced that $\graph(A)\in\nonegset$.
    Therefore, \Cref{th:gallai} guarantees that the system of inequalities~\eqref{eq:example_potential_inequalities} admits a real solution $x\in\R^\N$.
    To compute one, we apply the method shown in the proof.
    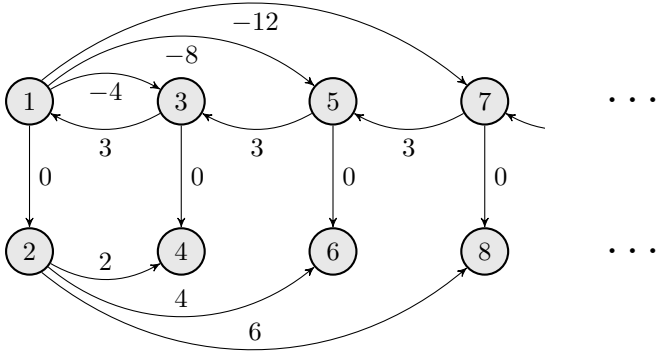
\begin{figure}[t]
        \centering
        \begin{tikzpicture}[on grid,node distance=2cm]

\newcommand{\myangle}{45}

\node [place] (p1) {1};
\node [place,right=of p1] (p2) {3};
\node [place,right=of p2] (p3) {5};
\node [place,right=of p3] (p4) {7};
\node [place,right=of p4] (p5) {8};

\node [place,below=of p1] (p11) {2};
\node [place,right=of p11] (p21) {4};
\node [place,right=of p21] (p31) {6};
\node [place,right=of p31] (p41) {8};
\node [right=of p41] (dots) {\Large$\cdots$};

\draw [-stealth'] (p1) to [bend left] node [below] {$-4$} (p2);
\draw [-stealth'] (p1) to [bend left=40] node [below] {$-8$} (p3);
\draw [-stealth'] (p1.60) to [bend left=40] node [below] {$-12$} (p4);

\draw [-stealth'] (p2) to [bend left] node [auto] {$3$} (p1);
\draw [-stealth'] (p3) to [bend left] node [auto] {$3$} (p2);
\draw [-stealth'] (p4) to [bend left] node [auto] {$3$} (p3);
\draw [-stealth'] (p5) to [bend left] node [auto] {$3$} (p4);

\fill [white] ($(p4)!0.4!(p5)+(0,-1cm)$) rectangle ($(p5)+(.5cm,.5cm)$);

\node [right=of p4] (dots1) {\Large$\cdots$};

\draw [-stealth'] (p1) to node [auto] {$0$} (p11);
\draw [-stealth'] (p2) to node [auto] {$0$} (p21);
\draw [-stealth'] (p3) to node [auto] {$0$} (p31);
\draw [-stealth'] (p4) to node [auto] {$0$} (p41);

\draw [-stealth'] (p11) to [bend right] node [auto] {$2$} (p21);
\draw [-stealth'] (p11) to [bend right=40] node [auto] {$4$} (p31);
\draw [-stealth'] (p11.-60) to [bend right=40] node [auto] {$6$} (p41);
\end{tikzpicture}
        \caption{Precedence graph corresponding to the difference bound matrix of \Cref{ex:potential_inequalities_1}.}\label{fi:graph_to_circledast}
    \end{figure}
    The Kleene star of $A$ is\footnote{In this example, elements of $A^*$ and $A^\circledast$ were obtained by inspection using their graphical interpretation.}
    \[
        A^* = 
        \begin{bNiceMatrix}
            0 &  \cdot &     3 &  \cdot &     6 &  \cdot &     9 &  \cdot & \cdots\\
            0 &     0 &     3 &  \cdot &     6 &  \cdot &     9 &  \cdot & \cdots\\
            -4 &  \cdot &     0 &  \cdot &     3 &  \cdot &     6 &  \cdot & \cdots\\
            2 &     2 &     5 &     0 &     8 &  \cdot &    11 &  \cdot & \cdots\\
            -8 &  \cdot &    -5 &  \cdot &     0 &  \cdot &     3 &  \cdot & \cdots\\
            4 &     4 &     7 &  \cdot &    10 &     0 &    13 &  \cdot & \cdots\\
            -12 &  \cdot &    -9 &  \cdot &    -6 &  \cdot &     0 &  \cdot & \cdots\\
            6 &     6 &     9 &  \cdot &    12 &  \cdot &    15 &     0 & \cdots\\
            \vdots &     \vdots &     \vdots &  \vdots & \vdots &  \vdots & \vdots & \vdots & \ddots
        \end{bNiceMatrix}.
    \]
    \zor{Let us define $\Phi$ according to the good order $\preceq$ derived from the Cantor pairing function $f:\N\times\N\rightarrow\N$.}
    The precedence graph $\graph(A^*)$ consists of infinitely many maximal strongly connected subgraphs and, in this particular example, the sequence $\Phi^k(A)$ does not converge to $A^\circledast$ in finitely many steps, but only for $k\rightarrow \infty$.
    In the first application of $\Phi$, \zor{the least (according to $\preceq$)} pair $(i,j)$ such that $(A^*)_{ij}=-\infty$ and $(A^*)_{ji}\neq -\infty$ is $(1,2)$.
    Matrix $\Phi(A)$ with this choice coincides with $A^*$ except for entry $(1,2)$, which is equal to $0 = -(A^*)_{21}$.
    After infinitely many applications of $\Phi$, we would get
    \[
        A^\circledast = 
        \begin{bNiceMatrix}
            0 &      0 &     3 &    -2 &     6 &    -4 &     9 &    -6 & \cdots\\ 
            0 &     0 &     3 &    -2 &     6 &    -4 &     9 &    -6 & \cdots\\
            -4 &    -4 &     0 &    -6 &     3 &    -8 &     6 &   -10 & \cdots\\
            2 &     2 &     5 &     0 &     8 &    -2 &    11 &    -4 & \cdots\\
            -8 &    -8 &    -5 &   -10 &     0 &   -12 &     3 &   -14 & \cdots\\
            4 &     4 &     7 &     2 &    10 &     0 &    13 &    -2 & \cdots\\
            -12 &   -12 &    -9 &   -14 &    -6 &   -16 &     0 &   -18 & \cdots\\
            6 &     6 &     9 &     4 &    12 &     2 &    15 &     0 & \cdots\\
      \vdots & \vdots & \vdots & \vdots & \vdots & \vdots & \vdots & \vdots & \ddots 
        \end{bNiceMatrix}.
    \]
    We can now extract a solution of~\eqref{eq:example_potential_inequalities} from any column of $A^\circledast$.
    For instance, the first column provides the solution (written in standard algebra)
    \[
        x_k = 
        \begin{dcases}
            -4\cdot\frac{k-1}{2} & \mbox{if $k$ is odd,}\\
            2\cdot\frac{k-2}{2} & \mbox{if $k$ is even.}
        \end{dcases}
    \]
    Note that the first columns of $A^*$ and $A^\circledast$ are identical; the reason is that there is a path from node $1$ to any other node in $\graph(A)$.
    This illustrates a case in which the computation of $A^\circledast$ for finding solutions of precedence constraints can be avoided.
\end{example}

The procedure described in the proof of~\Cref{th:gallai} has significant theoretical value.
However, in the case of infinite difference bound matrices it does not immediately translate into a useful algorithm.\footnote{Recall the difference between an algorithm and a procedure; only algorithms need to terminate in finite time.}
In~\Cref{se:P_TEGs}, we will study interesting classes of precedence constraints with infinitely many variable\zor{s} in infinitely many inequalities where verifying whether a solution exists in finite time is always possible.

\section{$\N$-periodic graphs}\label{se:N_periodic_graphs}

In this section, we study a decision problem related to $\N$-periodic graphs.

\subsection{Definitions}

Let us first define static (or uniform) graphs. %

\begin{definition}[Static graph]
    Given three $n\times n$ matrices $L$, $C$, $R$ (standing, respectively, for "left", "center", and "right") with elements from $\Rmax$, the associated static graph $\graph = \graph(L,C,R) = (\nodes,\arcs,w)$ is the weighted multi-directed graph\footnote{\zor{In multi-directed graphs, there may be multiple arcs connecting two nodes.}} with set of nodes $\nodes = \dint{1,n}$, set of arcs $\arcs\subseteq \nodes\times \{-1,0,+1\}\times \nodes$, and weight function $w:\arcs\rightarrow \R$, defined such that there is an arc $e = (i,s,j)\in \arcs$ from the upstream node $i \eqqcolon \myup(e)$ to the downstream node $j \eqqcolon \mydown(e)$ with shift $s \eqqcolon \myshift(e)$ and weight $(X_{s})_{ji} \eqqcolon w(e)$ if and only if $(X_{s})_{ji} \neq -\infty$, where
\[
    X_s = 
    \begin{dcases}
        L & \mbox{if }s = -1,\\
        C & \mbox{if }s = 0,\\
        R & \mbox{if }s = +1.\qedlineend
    \end{dcases}
\]
\end{definition}

Every static graph induces (or generates) $\Z$-periodic, $\N$-periodic, and -- more generally -- $\S$-periodic graphs as follows.

\begin{definition}[$\S$-periodic graphs]
Let $\S$ be any subset of $\Z$ and let $\graph = \graph(L,C,R) = (\nodes,\arcs,w)$ be a static graph.
The $\S$-periodic graph $\graph_\S = \graph_\S(L,C,R) = (\nodes_{\S},\arcs_{\S},w)$ induced by $\graph$ is the (possibly infinite) weighted directed graph with set of nodes $\nodes_{\S} = \nodes \times \S$, set of arcs 
\[
    \arcs_\S = \{((i,k),(j,k+s))\mid (i,s,j)\in \arcs,\ k,k+s\in\S\},
\] 
and weight function\footnote{With slight abuse of notation, we denote by $w$ the weight functions of both $\graph$ and $\graph_\S$.} $w:\arcs_\S\rightarrow \R$ defined by $w(((i,k),(j,k+s))) = w((i,s,j))$.
The base and shift\footnote{The notation on static and $\S$-periodic graphs differs among publications. For instance, what here is called "shift" is referred as to "transit" in \cite{hoefting1995minimum,orlin1984some} and "height" in \cite{munier2011graph}. Moreover, $\N$-periodic graphs are called uniform graphs in \cite{munier2011graph}.} of a node $v_\S = (i,k)\in \nodes_\S$ are, respectively, $\myheight(v_\S) \coloneqq i$ and $\myshift(v_\S) \coloneqq k$.
For every arc $e_\S = ((i,k),(j,k+s))\in \arcs_\S$, let $\myup(e_\S) \coloneqq (i,k)$, $\mydown(e_\S) \coloneqq (j,k+s)$, $\myshift(e_\S) \coloneqq s$.
\end{definition}

\begin{example}\label{ex:simple_graph}
Figure~\ref{fi:static_periodic_graphs} illustrates the static graph associated with matrices 
\[
L = \begin{bmatrix}
        \alpha & -\infty\\-\infty & -3
    \end{bmatrix}, \ C = \begin{bmatrix}
        -\infty & -\infty\\0 & -\infty
    \end{bmatrix}, \ R = \begin{bmatrix}
        \beta & -\infty\\-\infty & 2
    \end{bmatrix},
\]
where $\alpha,\beta\in\R$, together with (portions of) its induced $\N$-periodic and $\Z$-periodic graphs.
\end{example}

$\Z$-periodic graphs have been studied under the name of periodic (or dynamic) graphs, see, \eg \cite{orlin1984some,hoefting1995minimum}.
By relabeling the nodes of a $\Z$-periodic graph $\graph_\Z$, we can see that it coincides with the precedence graph $\graph(M_\Z)$ of matrix
\[
    M_\Z = 
    \begin{bmatrix}
        \ddots & \vdots & \vdots & \vdots & \vdots & \iddots\\
        \cdots & C & L & \cdot & \cdot & \cdots\\
        \cdots & R & C & L & \cdot & \cdots\\
        \cdots & \cdot & R & C & L & \cdots\\
        \cdots & \cdot & \cdot & R & C & \cdots\\
        \iddots & \vdots & \vdots & \vdots & \vdots & \ddots
    \end{bmatrix}\in\Rmax^{\Z\times \Z},
\]
where each "$\cdot$" stands for \zor{the} $n\times n$ matrix $\mathcal{E}$.
$\N$-periodic graphs have incidence matrix of the form
\begin{equation}\label{eq:M_N}
    M_\N = 
    \begin{bmatrix}
        C & L & \cdot & \cdot & \cdots\\
        R & C & L & \cdot & \cdots\\
        \cdot & R & C& L & \cdots\\
        \cdot & \cdot & R & C& \cdots\\
        \vdots & \vdots & \vdots & \vdots & \ddots
    \end{bmatrix}\in\Rmax^{\N\times \N}.
\end{equation}

\begin{figure}[t]
\centering
\subfloat[Static graph $\graph$.]{\label{fi:static_graph}\resizebox{!}{.18\textheight}{\begin{tikzpicture}[node distance=2cm and 2cm,place/.append style={minimum size=1cm},inner sep=1pt,on grid]

{
\footnotesize
\node [place] (ntop) at (0,0) {$1$};
\node [place] (nbot) at (0,2*-1) {$2$};
}
\draw [arc] (ntop) to node[auto] {$0,0$} (nbot);
\draw [arc] (ntop) to [out=20,in=-20,loop] node[above=.3cm] {$+1,\beta$} (ntop);
\draw [arc] (ntop) to [out=180-20,in=180+20,loop] node[above=.3cm] {$-1,\alpha$} (ntop);
\draw [arc] (nbot) to [out=20,in=-20,loop] node[below=.3cm] {$+1,2$} (nbot);
\draw [arc] (nbot) to [out=180-20,in=180+20,loop] node[below=.3cm] {$-1,-3$} (nbot);

\draw [arc,draw=none] (-.5,-3.5) to (.5,-3.2);

\end{tikzpicture}}}
\hfill
\subfloat[Portion of an $\N$-periodic graph $\graph_\N$.]{\label{fi:N_periodic_graph}\resizebox{!}{.2\textheight}{\begin{tikzpicture}[node distance=2cm and 2cm,place/.append style={minimum size=1cm},on grid,inner sep=1pt]

\clip (2.2*3+1.5,-3.5) rectangle (7*2.2-1.3,1);

\foreach [count=\i, evaluate=\i as \zo using int(\i-3)] \z in {1,...,7}
{
\ifnum\z<4
\else
    {
    \footnotesize
    \node [place] (ntop\z) at (2.2*\z,0) {$(1,\zo)$};
    \node [place] (nbot\z) at (2.2*\z,2*-1) {$(2,\zo)$};
    }
    \draw [arc] (ntop\z) to node[auto] {$0$} (nbot\z);
\fi
}

\foreach[count=\i, evaluate=\i as \zz using int(\i+4)] \z in {4,5,...,6}
{
\draw [arc] (ntop\z) to [bend left=30] node[auto] {$\beta$} (ntop\zz);
\draw [arc] (nbot\z) to [bend left=30] node[auto] {$2$} (nbot\zz);
}

\foreach[count=\i, evaluate=\i as \zz using int(\i+3)] \z in {5,6,...,7}
{
\draw [arc] (ntop\z) to [bend left=30] node[auto] {$\alpha$} (ntop\zz);
\draw [arc] (nbot\z) to [bend left=30] node[auto] {$-3$} (nbot\zz);
}

\draw [arc] (4*2.2,-3.2) to node[pos=.9,above] {$\N$} (7*2.2-1.5,-3.2);

\end{tikzpicture}}}
\hfill
\subfloat[Portion of a $\Z$-periodic graph $\graph_\Z$.]{\label{fi:Z_periodic_graph}
    \resizebox{!}{.2\textheight}{\begin{tikzpicture}[node distance=2cm and 2cm,arc/.style={->,>=stealth'},place/.style={draw,thick,fill=mygrey!15!white,circle,minimum size=1cm},on grid,inner sep=1pt]

\clip (2.2+1.3,-3.5) rectangle (7*2.2-1.3,1);

\foreach [count=\i, evaluate=\i as \zo using int(\i-4)] \z in {1,...,7}
{
{
\footnotesize
\node [place] (ntop\z) at (2.2*\z,0) {$(1,\zo)$};
\node [place] (nbot\z) at (2.2*\z,2*-1) {$(2,\zo)$};
}
\draw [arc] (ntop\z) to node[auto] {$0$} (nbot\z);
}

\foreach[count=\i, evaluate=\i as \zz using int(\i+1)] \z in {1,...,6}
{
\draw [arc] (ntop\z) to [bend left=30] node[auto] {$\beta$} (ntop\zz);
\draw [arc] (nbot\z) to [bend left=30] node[auto] {$2$} (nbot\zz);
}

\foreach[count=\i, evaluate=\i as \zz using int(\i)] \z in {2,...,7}
{
\draw [arc] (ntop\z) to [bend left=30] node[auto] {$\alpha$} (ntop\zz);
\draw [arc] (nbot\z) to [bend left=30] node[auto] {$-3$} (nbot\zz);
}

\draw [stealth'-stealth'] (2.2+1.5,-3.2) to node[pos=.95,above] {$\Z$} (7*2.2-1.5,-3.2);

\end{tikzpicture}}
}
\caption{Static graph and corresponding $\N$- and $\Z$-periodic graphs. Every arc $e$ in the static graph is labeled "$\myshift(e),w(e)$". Every node $v$ in the periodic graphs is labeled "$(\myheight(v),\myshift(v))$".}
\label{fi:static_periodic_graphs}
\end{figure}
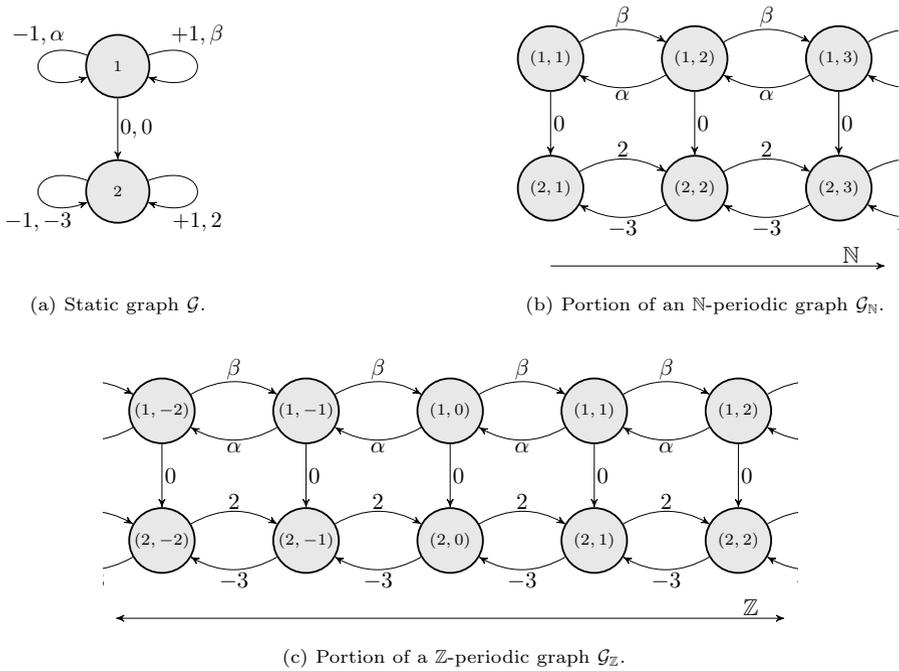

\zor{\begin{example}
Observe that the graph in \Cref{fi:infinite_precedence_graph} is $\N$-periodic, while the one in \Cref{fi:graph_to_circledast} is not.
\end{example}}

To simplify the notation in the propositions of this section, it is convenient to provide alternative definitions of paths and circuits for static and $\S$-periodic graphs.
A path $p$ in either a static or an $\S$-periodic graph is an alternating sequence $p = (v_1,e_1,v_2,\ldots,v_m)$ of nodes $v_i$ and arcs $e_i$ such that $\myup(e_i) = v_i$ and $\mydown(e_i) = v_{i+1}$ for all $i\in\dint{1,m-1}$.
The length of $p$ is $\mylen(p) = m-1$.
A path $p$ is called circuit if $v_1\eqqcolon\myup(p)$ and $v_m\eqqcolon\mydown(p)$ coincide. %
A path $p$ in an $\S$-periodic graph is a pseudo-circuit if $\myheight(v_1) = \myheight(v_m)$.
Each node (resp., arc, path, pseudo-circuit) of an $\S$-periodic graph $\graph_\S$ corresponds to a unique node (resp., arc, path, circuit) of the associated static graph $\graph$.
On the other hand, each node (resp., arc, path, circuit) of $\graph$ induces infinitely many nodes (resp., arcs, paths, pseudo-circuits) of $\graph_\N$ and $\graph_\Z$.
The shift, weight, left-shift ($\mylshift$), and right-shift ($\myrshift$) of path $p$ are defined by\footnote{In the definition of $\mylshift$ and $\myrshift$, we assume by convention that the empty sum is equal to $0$.}
\[
    \myshift(p) = \sum_{i=1}^{m-1} \myshift(e_i),\quad
    \myweight(p) = \sum_{i=1}^{m-1} w(e_i),
\]
\[
    \mylshift(p) = \min_{i\in\dint{0,m-1}} \sum_{j=1}^{i} \myshift(e_j)\leq 0,\
    \myrshift(p) = \max_{i\in\dint{0,m-1}} \sum_{j=1}^i \myshift(e_j)\geq 0.
\]
Furthermore, if $p$ is a path in $\graph_\Z$, then it is also a path in $\graph_\N$ if and only if
\begin{equation}\label{eq:path_condition}
    \myshift(\myup(p))+\mylshift(p)\in\N.
\end{equation}
Let $p_1,p_2,\ldots,p_m$ be paths such that $\mydown(p_i) = \myup(p_{i+1})$ for all $i\in\dint{1,m-1}$.
Then we write $p_1p_2\cdots p_m$ to indicate the path obtained by concatenating $p_1,p_2,\ldots,p_m$.
If $p$ is a circuit and $x\in\N$, we define $p^x = p p^{x-1}$, where $p^0$ indicates the empty path, which has zero length.

\subsection{Detecting $\infty$-weight paths}\label{su:bounds}

Consider the following decision problem.\\[1.2mm]
\noindent
$\infty$-WEIGHT $\S$-PATH\\
\textbf{Instance:} Matrices $L,C,R\in\Rmax^{n\times n}$.\\
\textbf{Question:} Does $\graph_\S(L,C,R)$ contain an $\infty$-weight path?

\begin{example}\label{ex:simple_graph_2}
As an illustrative example, take the graphs of Figure~\ref{fi:static_periodic_graphs}.
For values $\alpha=-1,\beta=2$, both $\graph_\N$ and $\graph_\Z$ contain $\infty$-weight paths, as there exists a circuit with positive weight with source node $(1,k)$, for all $k$. %
Unlike finite graphs, however, $\S$-periodic graphs may contain $\infty$-weight paths even when there are no positive-weight circuits.
For example, this is the case for $\graph_\N$ and $\graph_\Z$ when $\alpha=-5,\beta=4$, as both of them contain an $\infty$-weight path from node $(1,k)$ to node $(2,k)$, for all $k$.
On the other hand, when $\alpha=-1,\beta=1$, only $\graph_\Z$ contains $\infty$-weight paths, each corresponding to a sequence of paths $p^k_1,p^k_2,\ldots$ with increasing weight from node $(1,k)$ to node $(2,k)$.
The same sequence cannot be built in $\graph_\N$, since for all $k\in\N$ and for $h$ large enough, $p^k_h$ does not satisfy~\eqref{eq:path_condition}, and thus is not a path in $\graph_\N$.
\end{example}

In \cite[Thorem 4.8]{hoefting1995minimum}, a polynomial-time algorithm that solves $\infty$-WEIGHT $\Z$-PATH was presented.
Munier Kordon solved the problem $\infty$-WEIGHT $\N$-PATH in weakly polynomial time \cite{munier2011graph}.\footnote{\zor{In \cite{munier2011graph}, it is assumed that, in the considered $\N$-periodic graph, there always exists a path with positive weight from any node $(i,k)$ to node $(i,k+h)$ for all $h\in\N$. This condition, however, does not seem to be essential for the algorithm proposed there.}}
In the following, we prove that any instance of $\infty$-WEIGHT $\N$-PATH can be solved in strongly polynomial time.

\begin{restatable}{lemma}{firstlemma}\label{le:length_supremal_path_N_periodic_graph}
Let $\graph$ be a static graph with $n$ nodes and let $i,j\in\dint{1,n}$.
Suppose that $\graph_\N$ does not contain $\infty$-weight paths from node $(i,1)$ to node $(j,1)$.
Then, the maximal weight of all paths from node $(i,1)$ to node $(j,1)$ in $\graph_\N$ is attained by a path with right-shift at most $n^2$.
\end{restatable}

The proof of \Cref{le:length_supremal_path_N_periodic_graph} can be found in \Cref{se:prooffirstlemma}.

Let us denote by $\Pi_{ij}(h)$ the supremal weight of all paths from node $(j,1)$ to node $(i,1)$ of right-shift at most $h\in\No$ in $\graph_{\N}$:
\[
    \begin{array}{rcl}
        \Pi_{ij}(h) &=& \sup \{\myweight(p)\mid p \mbox{ is a path in $\graph_{\N}$ from $(j,1)$ to $(i,1)$}\\
                    && \phantom{\sup \ \ \myweight(p)\mid} \mbox{and } \myrshift(p)\leq h\}\\
                    &=& \sup \{\myweight(p)\mid p \mbox{ is a path in $\graph_{\dint{1,h+1}}$ from $(j,1)$ to $(i,1)$}\}.
    \end{array} 
\] 
By definition,
\[
    \Pi_{ij}(+\infty) \coloneqq \lim_{h\rightarrow +\infty} \Pi_{ij}(h)
\] 
is the supremal weight of all paths in $\graph_{\N}$ from $(j,1)$ to $(i,1)$.
For all $h\in\No\cup\{+\infty\}$, we can construct matrix $\Pi(h)\in\Rbar^{n\times n}$ with $(i,j)$-element $(\Pi(h))_{ij} = \Pi_{ij}(h)$.
Note that $\Pi_{ij}(h)=-\infty$ if and only if there is no path from $(j,1)$ to $(i,1)$ with right-shift at most $h$.
Moreover, $\Pi_{ij}(h) = +\infty$ if and only if there is an $\infty$-weight path with right-shift at most $h$; as seen in \Cref{ex:simple_graph_2}, it is possible that $\Pi_{ij}(h)\in\Rmax$ for all $h\in\No$ and $\Pi_{ij}(+\infty) = +\infty$.
Using the max-plus algebra, we can get the following recursive formula for $\Pi(h)$ (for the proof, see \Cref{se:proofsecondlemma}).\footnote{\zor{A proof of \Cref{le:formula_Pi} can also be found in the technical report \cite{declerck2009extremal}.}}

\begin{restatable}{lemma}{secondlemma}\label{le:formula_Pi}
For all $h\in\No$, 
\[
    \begin{array}{rcl}
        \Pi(0) &=& C^+,\\
        \Pi(h+1) &=& (L \Pi(h)^* R \oplus C)^+.\qedlineend
    \end{array} 
\]
\end{restatable}

In algebraic terms, \Cref{le:length_supremal_path_N_periodic_graph} showed that, if $\graph_{\N}$ contains no $\infty$-weight paths from node $(i,1)$ to node $(j,1)$ for all $i,j\in\dint{1,n}$, then
\[
    \Pi(+\infty) = \Pi(n^2)\in\Rmax^{n\times n}.
\]
If $\infty$-weight paths from $(i,1)$ to $(j,1)$ exist for some $i,j$, then either $\Pi(n^2)\not\in\Rmax^{n\times n}$ (\ie there are circuits with positive weight and right-shift at most $n^2$), or we must have $\Pi(n^2+1)\neq\Pi(n^2)$.
If indeed the latter inequation were not true, then from the formula in \Cref{le:formula_Pi} we would have found a fixed point $\Pi(n^2)$ of mapping $F(\Pi) = (L\Pi^*R\oplus C)^+$, implying that $\Pi(n^2) = \Pi(n^2+1) = \Pi(n^2+2) =\dots = \Pi(+\infty)$.
This shows that \Cref{le:length_supremal_path_N_periodic_graph} gives a necessary and sufficient condition for the existence of $\infty$-weight paths between any nodes $(i,1)$ and $(j,1)$, where $i,j\in\dint{1,n}$.
In fact, in \Cref{se:proofthirdlemma} we prove the following, stronger result.

\begin{restatable}{lemma}{thirdlemma}\label{le:inf_weight_paths_N_periodic}
    Let $\graph$ be a static graph with $n$ nodes.
    Then $\graph_\N$ does not contain $\infty$-weight paths if and only if 
    \[
        \Pi(n^2+1) = \Pi(n^2) \mbox{ and }\Pi(n^2)\in\Rmax^{n\times n}.\qedlineend
    \]
\end{restatable} 

Because of \Cref{le:inf_weight_paths_N_periodic}, we can now state the main theorem of this section.
Recall that, given two matrices $A,B\in\Rmax^{n\times n}$, computing $A\oplus B$ and $A\otimes B$ using naive algorithms requires, respectively, $O(n^2)$ and $O(n^3)$ operations.\footnote{We assume to be working in the arithmetic model of computation, where standard operations between reals can be performed in constant time. All the time complexities given in this paper apply also in the Turing machine model, if reals are substituted by integers or rational numbers.}
Moreover, using the Floyd-Warshall algorithm, it is possible to verify the existence of circuits with positive weight in a graph $\graph(A)$ with $n$ nodes, and, if no such circuits exist, to compute $A^*$ in time $O(n^3)$ (see, e.g., \cite{cormen2022introduction}).

\begin{theorem}\label{th:polynomial_complexity}
The problem $\infty$-WEIGHT $\N$-PATH is solvable in strongly polynomial time complexity $O(n^5)$.
\end{theorem}
\begin{proof}
Given matrices $L,C,R\in\Rmax^{n\times n}$, we need to decide whether $\graph_\N = \graph_{\N}(L,C,R)$ contains an $\infty$-weight path.
According to \Cref{le:inf_weight_paths_N_periodic}, this can be done by computing $\Pi(0),\Pi(1),\dots,\Pi(n^2+1)$ and verifying, each time, whether $\Pi(h)\in\Rmax^{n\times n}$.
If for any $h\in\dint{0,n^2+1}$ we get $\graph(\Pi(h))\not \in\nonegset$, then the algorithm can be terminated as $\graph_{\N}$ contains a positive-weight circuit with shift at most $h$.
Otherwise, we need to check if $\Pi(n^2+1) = \Pi(n^2)$; the equality holds if and only if there are no $\infty$-weight paths in $\graph_\N$.

The whole procedure requires to compute at most $O(n^2)$ multiplications, additions, and Kleene stars on $n\times n$ matrices.
This results in an algorithm that terminates after $O(n^2n^3) = O(n^5)$ operations.
\end{proof}

\begin{example}\label{ex:simple_graph_3}
From the discussion carried out in Example~\ref{ex:simple_graph_2}, we know that $\graph_\N$ from \Cref{fi:static_periodic_graphs} contains an $\infty$-weight path only for certain parameters $\alpha,\beta$.
We can now verify this using \Cref{th:polynomial_complexity}.

For values $\alpha=-1,\beta=2$, it can be checked that $\graph(\Pi(0)) = \graph(C^+)\in\nonegset$, but $\graph(\Pi(1)) = \graph((LC^*R\oplus C)^+)\not\in \nonegset$; therefore, there is a positive-weight circuit with right-shift $1$. 
When $\alpha=-5,\beta=4$, observe that
\[
    \begin{bmatrix}
        -1 & -\infty\\5 & -1
    \end{bmatrix} = 
    \Pi(5) = \Pi(n^2+1) \neq \Pi(n^2) = \Pi(4) = 
    \begin{bmatrix}
        -1 & -\infty\\4 & -1
    \end{bmatrix}.
\]
Since $(\Pi(5))_{21} > (\Pi(4))_{21}$, there is an $\infty$-weight path starting from node $(1,1)$ and terminating in node $(2,1)$.
This means either that there is a positive-weight circuit with right-shift at least $6$ or that no positive-weight circuit exists but that there is an infinite sequence of elementary paths with infinite limit-weight.
In the considered case, it can be checked using the techniques presented in \cite{hoefting1995minimum,ZORZENON202219} that no positive-weight circuit exists, confirming that we are in the latter scenario.
On the other hand, when $\alpha=-1,\beta=1$, we have
\[
    \Pi(5) = 
    \Pi(4) = 
    \begin{bmatrix}
        0 & -\infty\\0 & -1
    \end{bmatrix}\in\Rmax^{2\times 2}.
\]
Thus, no $\infty$-weight path is present in this case.
\end{example}

\zor{\begin{remark}
We remark that the sequence of matrices $\{\Pi(h)\}_{h\in\No}$ defined in \Cref{le:formula_Pi} was first studied in \cite{5628259}.
What distinguishes our results from the one in \cite{5628259} is \Cref{le:inf_weight_paths_N_periodic}, which shows that computing the matrices up to $h=n^2+1$ is sufficient to decide the convergence of the sequence. 
\end{remark}}

\section{Ultimately periodic graphs}\label{se:ultimately_periodic_graphs}

In this section, we extend the results from \Cref{se:N_periodic_graphs} to the class of ultimately periodic graphs.

\subsection{Definitions}\label{su:definitions_ultimately_periodic_graphs}

Let $\graph_\wP = \graph(L_\wP,C_\wP,R_\wP) = (\nodes,\arcs_\wP,w_\wP)$ and $\graph_\wN = \graph(L_\wN,C_\wN,R_\wN) = (\nodes,\arcs_\wN,w_\wN)$ be two static graphs with the same set of nodes $\nodes = \dint{1,n}$, and let $C_\wT\in\Rmax^{n\times n}$ be a matrix.
The \emph{ultimately periodic graph} induced by $\graph_\wN$, $C_\wT$, and $\graph_\wP$ is the infinite weighted digraph $\graph_{\textup{U}} = \graph_{\textup{U}}(\graph_\wN,C_\wT,\graph_\wP) = (\nodes_{\textup{U}},\arcs_{\textup{U}},w_{\textup{U}})$ with set of nodes $\nodes_{\textup{U}} = \nodes\times \Z$, set of arcs\footnote{\zor{Note that, if $(i,s,j)\in\arcs_\wN\cup\arcs_\wP$, then $s\in\{-1,0,+1\}$. For this reason, the three conditions in the definitions of $\arcs_{\textup{U}}$ and $w_{\textup{U}}$ are mutually exclusive.}}
\[
    \arcs_{\textup{U}} = \left\{((i,k),(j,k+s))\ \left| \
        \begin{array}{cl}
                    &((k<0 \vee k+s<0) \wedge (i,s,j)\in \arcs_\wN) \\
\vee&((k>0 \vee k+s > 0) \wedge (i,s,j)\in \arcs_\wP) \\
\vee&(k = s = 0  \wedge (C_{\wT})_{ji}\neq -\infty)
        \end{array} 
        \right.
    \right\},
\]
and weight function $w_{\textup{U}}:\arcs_{\textup{U}}\rightarrow \R$ defined by
\[
    w_{\textup{U}}(((i,k),(j,k+s))) = 
    \begin{dcases}
        w_\wN((i,s,j)) & \mbox{if } k<0 \vee k+s < 0,\\
        w_\wP((i,s,j)) & \mbox{if } k>0 \vee k+s > 0,\\
        (C_\wT)_{ji} & \mbox{if } k=s=0.
    \end{dcases}
\]
Ultimately periodic graphs are the union of two periodic graphs -- a $\Z_{<0}$-periodic graph induced by $\graph_\wN$, called \emph{negative periodic part}, for nodes with negative shift and an $\N$-periodic graph induced by $\graph_\wP$, called \emph{positive periodic part}, for nodes with positive shift -- joined together via a finite graph $\graph(C_\wT)$ at zero shift called \emph{transient part}.
They are more general than $\Z$- and $\N$-periodic graphs, as ultimately periodic graphs with $X_\wN=X_\wP$ for all $X\in\{R,L\}$ and $C_\wN=C_\wT=C_\wP$ are $\Z$-periodic, and those with $X_\wN=\mathcal{E}$ for all $X\in\{R,L,C\}$ and $C_\wT=C_\wP$ are $\N$-periodic (except for having infinitely many additional nodes not connected to other nodes).
The definitions of paths and circuits are derived from those of $\S$-periodic graphs.
Each ultimately periodic graph can also be thought of as the precedence graph of matrix
\begin{equation}\label{eq:M_U}
    M_\textup{U} =
    \begin{bmatrix}
        \ddots & \vdots & \vdots & \vdots & \vdots & \vdots & \vdots & \vdots & \iddots\\
        \cdots & C_\wN & L_\wN & \cdot & \cdot &\cdot & \cdot & \cdot & \cdots\\
        \cdots & R_\wN & C_\wN & L_\wN & \cdot &\cdot &\cdot & \cdot & \cdots\\
        \cdots & \cdot & R_\wN & C_\wN & L_\wN & \cdot &\cdot &\cdot &\cdots\\
        \cdots & \cdot & \cdot & R_\wN & C_\wT & L_\wP & \cdot &\cdot &\cdots\\
        \cdots & \cdot & \cdot & \cdot & R_\wP & C_\wP & L_\wP &\cdot &\cdots\\
        \cdots & \cdot&\cdot & \cdot &\cdot & R_\wP & C_\wP & L_\wP &\cdots\\
        \cdots & \cdot &\cdot & \cdot &\cdot & \cdot &R_\wP & C_\wP & \cdots\\
        \iddots & \vdots & \vdots & \vdots & \vdots& \vdots& \vdots & \vdots & \ddots
    \end{bmatrix}\in\Rmax^{\Z\times \Z}.
\end{equation}

\begin{example}\label{ex:ultimately_periodic_graph}
    \begin{figure}[t]
        \centering
        \begin{tikzpicture}[node distance=2cm and 2cm,place/.append style={minimum size=.3cm},on grid]

\foreach \z in {1,...,9}
{
    \coordinate (n1\z) at (1.5*\z,0) {};
    \coordinate (n2\z) at (1.5*\z,-1) {};
    \coordinate (n3\z) at (1.5*\z,-2) {};
    \coordinate (n4\z) at (1.5*\z,-3) {};
}

\clip ($(n41)!.5!(n42)+(-.3,-1.5)$) rectangle ($(n18)!.5!(n19)+(.3,.3)$);

\foreach \z in {1,...,9}
{
    {
    \node [place] (n1\z) at (n1\z) {};
    \node [place] (n2\z) at (n2\z) {};
    \node [place] (n3\z) at (n3\z) {};
    \node [place] (n4\z) at (n4\z) {};
    }
}

\foreach \z in {1,...,4}
{
    \draw [arc] (n4\z) to [bend left=30]  (n1\z);
}
\draw [arc] (n15) to node [left] {$1$} (n25);
\foreach \z in {6,...,9}
{
    \draw [arc] (n3\z) to (n4\z);
}

\foreach[evaluate=\z as \zz using int(\z+1)] \z in {1,...,4}
{
    \draw [arc] (n1\z) to (n1\zz);
}
\foreach[evaluate=\z as \zz using int(\z+1)] \z in {5,...,8}
{
    \draw [arc] (n1\z) to (n2\zz);
    \draw [arc] (n2\z) to (n1\zz);
}

\foreach [evaluate=\z as \zz using int(\z-1)] \z in {2,...,5}
{
\draw [arc] (n4\z) to (n4\zz);
}
\foreach [evaluate=\z as \zz using int(\z-1)] \z in {6,...,9}
{
\draw [arc] (n2\z) to node [pos=.8,above left] {$1$} (n3\zz);
\draw [arc] (n4\z) to node [below] {$-1$} (n4\zz);
}

\draw [stealth'-stealth'] ($(n42)+(-.8,-.7cm)$) to node[pos=.97,below] {$\Z$} ($(n48)+(.8,-.7cm)$);
\foreach [evaluate=\z as \zz using int(\z-5)] \z in {1,...,9}
{
\draw ($(n4\z)+(0,-.8cm)$) to node[pos=0,below] {$\zz$} ($(n4\z)+(0,-.6cm)$);
}

\draw [glow=myred,draw=none] (n15) to (n25) to (n16) to (n27) to (n36) to (n46) to (n44) to [bend left=30] (n14) to (n14) to (n15);

\end{tikzpicture}
        \caption{Ultimately periodic graph for \Cref{ex:ultimately_periodic_graph}. The weight of arcs is $0$ unless indicated otherwise. The highlighted arcs form a positive-weight circuit.}\label{fi:ultimately_periodic_graph}
    \end{figure}
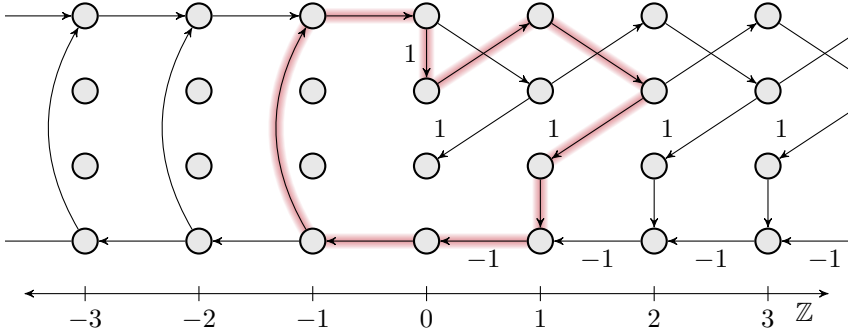
    The ultimately periodic graph corresponding to the $4\times 4$ matrices
    \[L_\wN = 
    \begin{bmatrix}
        \cdot & \cdot &\cdot &\cdot \\ 
        \cdot & \cdot &\cdot &\cdot \\ 
        \cdot & \cdot &\cdot &\cdot \\ 
        \cdot & \cdot &\cdot &0
    \end{bmatrix},\
    R_\wN = 
    \begin{bmatrix}
        0 & \cdot &\cdot &\cdot \\ 
        \cdot & \cdot &\cdot &\cdot \\ 
        \cdot & \cdot &\cdot &\cdot \\ 
        \cdot & \cdot &\cdot &\cdot 
    \end{bmatrix},\
    C_\wN = 
    \begin{bmatrix}
        \cdot & \cdot &\cdot &0 \\ 
        \cdot & \cdot &\cdot &\cdot \\ 
        \cdot & \cdot &\cdot &\cdot \\ 
        \cdot & \cdot &\cdot &\cdot 
    \end{bmatrix},\
    C_\wT = 
    \begin{bmatrix}
        \cdot & \cdot &\cdot &\cdot \\ 
        1 & \cdot &\cdot &\cdot \\ 
        \cdot & \cdot &\cdot &\cdot \\ 
        \cdot & \cdot &\cdot &\cdot 
    \end{bmatrix},
\]
\[    
    L_\wP = 
    \begin{bmatrix}
        \cdot & \cdot &\cdot &\cdot \\ 
        \cdot & \cdot &\cdot &\cdot \\ 
        \cdot & 1 &\cdot &\cdot \\ 
        \cdot & \cdot &\cdot &-1
    \end{bmatrix},\
    R_\wP = 
    \begin{bmatrix}
        \cdot & 0 &\cdot &\cdot \\ 
        0 & \cdot &\cdot &\cdot \\ 
        \cdot & \cdot &\cdot &\cdot \\ 
        \cdot & \cdot &\cdot &\cdot 
    \end{bmatrix},\
    C_\wP = 
    \begin{bmatrix}
        \cdot & \cdot &\cdot &\cdot \\ 
        \cdot & \cdot &\cdot &\cdot \\ 
        \cdot & \cdot &\cdot &\cdot \\ 
        \cdot & \cdot &0 &\cdot 
    \end{bmatrix}
\]
    is represented in \Cref{fi:ultimately_periodic_graph}.
\end{example}

As for $\S$-periodic graphs, we can consider the following decision problem.\\[1.2mm]
\noindent
$\infty$-WEIGHT ULTIMATE-PATH\\
\textbf{Instance:} Matrices $L_\wN,C_\wN,R_\wN,C_\wT,L_\wP,C_\wP,R_\wP\in\Rmax^{n\times n}$.\\
\textbf{Question:} Does $\graph_{\textup{U}}(\graph_\wN,C_\wT,\graph_\wP)$, where $\graph_\wN = \graph(L_\wN,C_\wN,R_\wN)$ and $\graph_\wP = \graph(L_\wP,C_\wP,R_\wP)$, contain an $\infty$-weight path?

\subsection{Detecting $\infty$-weight paths}\label{su:longest_path_ultimately_periodic}

Define the mappings $\Pi_\wN,\Pi_\wP:\No\rightarrow\Rbar^{n\times n}$ recursively by
\[
    \begin{array}{rclrcl}
        \Pi_\wN(0) &=& C_\wN^+, & \Pi_\wP(0) &=& C_\wP^+,\\
         \Pi_\wN(h+1) &=& (R_\wN \Pi_\wN(h)^* L_\wN \oplus C_\wN)^+,& \quad
        \Pi_\wP(h+1) &=& (L_\wP \Pi_\wP(h)^* R_\wP \oplus C_\wP)^+.       
    \end{array} 
\]
The following lemma is proven in \Cref{se:prooffourthlemma}.

\begin{restatable}{lemma}{fourthlemma}\label{pr:infinite_weight_ultimately_periodic}
Let $\graph_{\textup{U}}$ be an ultimately periodic graph defined as above.
Then, $\graph_{\textup{U}}$ does not contains $\infty$-weight paths if and only if the following three conditions are satisfied:
    \begin{itemize}
    \item $\graph_\N(R_\wN,C_\wN,L_\wN)$ does not contain $\infty$-weight paths,
\item $\graph_\N(L_\wP,C_\wP,R_\wP)$ does not contain $\infty$-weight paths,
        \item $\graph(R_\wN \Pi_\wN(n^2)^* L_\wN \oplus C_\wT\oplus L_\wP \Pi_\wP(n^2)^* R_\wP)\in\nonegset$.\qedhere
    \end{itemize}
\end{restatable}

From the latter lemma, we can obtain an algorithm that checks the presence of $\infty$-weight paths and compute the supremal weight of paths in ultimately periodic graphs.
The algorithm has the same asymptotic time-complexity \zor{as} the one discussed in \Cref{su:bounds}.
This implies the following theorem.

\begin{theorem}\label{th:ultimately_periodic}
The problem $\infty$-WEIGHT ULTIMATE-PATH is solvable in strongly polynomial time complexity $O(n^5)$.
\end{theorem}

\begin{example}
    Take again the ultimately periodic graph from \Cref{ex:ultimately_periodic_graph}.
    It is possible to verify that $\Pi_\wN(16) = \Pi_\wN(17)\in\Rmax^{4\times 4}$ and $\Pi_\wP(16) = \Pi_\wP(17)\in\Rmax^{4\times 4}$.
    Therefore, neither the negative nor the positive periodic part of the graph contain $\infty$-weight paths.
    However, we get
    \[
        \graph(R_\wN \Pi_\wN(16)^* L_\wN \oplus C_\wT\oplus L_\wP \Pi_\wP(16)^* R_\wP ) = 
        \graph\left(\begin{bmatrix} 
            \cdot &\cdot &\cdot &0 \\
            1 &\cdot &\cdot &\cdot \\
            1 &\cdot &\cdot &\cdot \\
            -1 &0 &\cdot &\cdot 
    \end{bmatrix} \right)
        \not\in \nonegset.
    \]
    Thus, the graph contains a positive-weight circuit (highlighted in \Cref{fi:ultimately_periodic_graph}) visiting at least one node with zero shift.
\end{example}

\section{Consistency of P-time event graphs}\label{se:P_TEGs}

In this section we use \Cref{th:gallai}, \Cref{th:polynomial_complexity}, and \Cref{th:ultimately_periodic} for the analysis of P-time event graphs.
We start by recalling their definition and dynamics.

\subsection{P-time event graphs}

\begin{definition}[\cite{khansa1996p}]\label{de:PTPN}
An ordinary \textit{P-time Petri net} is a 5-tuple $(\places,\transitions,\arcs,\marking,\iota)$, in which $\places$ is a finite set of places, $\transitions$ is a finite set of transitions, $\arcs\subseteq (\places\times \transitions)\cup (\transitions \times \places)$ is the set of arcs connecting places to transitions and transitions to places, and $\marking:\places\rightarrow\No$ and $\iota:\places\rightarrow\{[\tau^-,\tau^+]\cap \R \mid \tau^-\in \R_{\geq 0},\tau^+\in\R_{\geq 0}\cup\{\infty\}\}$ are two maps that associate to each place $p\in\places$, respectively, its initial number of tokens (or marking) $\marking(p)$, and a time interval $\iota(p)=[\tau_p^-,\tau_p^+]\cap \R$. 
\end{definition}

The dynamics of ordinary P-time Petri nets evolves as follows.
A transition $t\in\transitions$ is said to be enabled if either it has no upstream places (i.e., $\forall p\in\places$, $(p,t)\notin\arcs$) or each upstream place $p\in \places$ contains at least one token that has resided in $p$ for a time included in interval $[\tau_p^-,\tau_p^+]\cap\R$.
Note that this time interval is always closed, unless $\tau_p^+=+\infty$, in which case it is of the form $[\tau_p^-,+\infty)$.
When transition $t$ is enabled, it can fire, causing one token to be instantaneously removed from each upstream place and one token to be instantaneously added to each downstream place.
If a token resides for too long in $p$, violating the constraint imposed by interval $[\tau_p^-,\tau_p^+]\cap \R$, then the token is said to be \textit{dead}.

In this paper, we focus on a subclass of P-time Petri nets called \textit{P-time event graphs} (P-TEGs).
A P-TEG is a P-time Petri net where each place has exactly one upstream and one downstream transition (\ie $\forall p\in \places$, $\exists ! (t_{\textup{up}},t_{\textup{down}})\in\transitions\times \transitions$ such that $(t_{\textup{up}},p)\in\arcs$ and $(p,t_{\textup{down}})\in\arcs$).

\begin{example}[Heat treatment unit]\label{ex:heat_treatment}
\begin{figure}[h]
    \centering
    \resizebox{.8\textwidth}{!}{
    \begin{tikzpicture}[place/.append style={minimum size=6mm},node distance=.2cm and 1.5cm]

\node [transition,label=below:$t_1$] (x1) {};
\zor{
\node [place,tokens=1,label=below:{$p_{21}: [2,3]$}] (p21) [right= of x1] {};
}
\node [place,tokens=0,label=above:{$p_{12}: [0,+\infty)$}] (p12) [above= of p21] {};
\node [transition,label=below:$t_2$] (x2) [right=of p21] {};
\node (p22) [below= of x2] {};
\node [place,tokens=0,label=below:{$p_{32}: [0.5,+\infty)$}] (p32) [right= of x2] {};
\node [place,tokens=1,label=above:{$p_{23}: [0.5,+\infty)$}] (p23) [above= of p32] {};
\node [transition,label=below:$t_3$] (x3) [right=of p32] {};
\node [place,tokens=1,label=below:{$p_{33}: [0,4]$}] (p33) [right= of x3] {};
\zor{
\node [place,tokens=1,label=below:{$p_{31}: [6,+\infty)$}] (p31) [below= of p22] {};
}

\draw [arc] (x1) to (p21);
\draw [arc] (p21) to (x2);
\draw [arc] (x2) to (p32);
\draw [arc] (p32) to (x3);
\draw [arc] (x1.-90+20) to [out=-60,in=180] (p31);
\draw [arc] (p31) to [out=0,in=-180+60] (x3.-90-20);
\draw [arc] (x2.90+20) to [out=90+45,in=0] (p12);
\draw [arc] (p12) to [out=180,in=45] (x1.90-20);
\draw [arc] (x3.90+20) to [out=90+45,in=0] (p23);
\draw [arc] (p23) to [out=180,in=45] (x2.90-20);
\draw [arc] (x3) to [bend right=30] (p33);
\draw [arc] (p33) to [bend right=30] (x3);
\end{tikzpicture}
    }
    \caption{P-TEG representing the heat treatment unit of \Cref{ex:heat_treatment}.}\label{fi:TEG_furnace}
\end{figure}
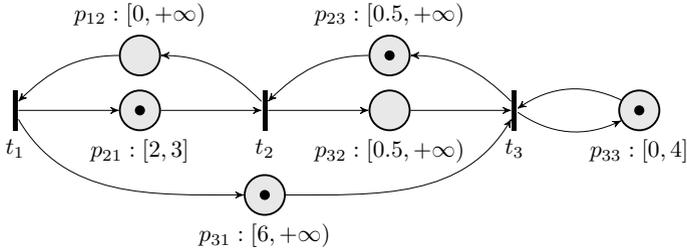
An example of P-TEG representing a heat treatment unit\zor{, consisting in a furnace and an autonomous guided vehicle (AGV),} is shown in \Cref{fi:TEG_furnace}.
Each firing of transition $t_1$ represents the arrival of a piece in \zor{the furnace.}
Every piece needs to be heated in \zor{the furnace} between $2$ and $3$ time units (see place $p_{21}$) to acquire the desired material properties.
\zor{The furnace has} capacity one (see place $p_{12}$), \ie \zor{it is} able to process one piece at a time.
The firing of transition $t_2$ indicates that one of the pieces has left \zor{the} furnace and has been loaded into \zor{the AGV}.
\zor{The AGV has} unitary capacity, and its purpose is to transport processed pieces to the next processing stage (not represented by the P-TEG) and then come back to the furnace.
A one-way journey of the AGV takes $0.5$ time units (see places $p_{32}$ and $p_{23}$).
The unloading of a piece from the AGV is indicated by the firing of transition $t_3$.
Customers require that one processed piece arrives at the next processing stage at most every $4$ time units (see place $p_{33}$).
Moreover, it is additionally required that pieces remain in the heat treatment unit (\ie furnace + AGV) at least for $6$ time units, in order to synchronize with other stages (see place $p_{31}$).
It is supposed that initially \zor{one part is} already being processed in the furnace, and that the AGV is ready to collect pieces from the furnace; this is reflected in the initial marking of the P-TEG.
\end{example}

We say that a P-TEG is \emph{consistent} if there exists an infinite sequence of firings of its transitions that does not cause any token death.
In order to study this property, it is convenient to state the dynamics of P-TEGs as a system of inequalities.
Before doing this, we need to clarify the role of initial conditions.

\subsection{Initial conditions}

Depending on the P-TEG application, different types of initial conditions may be considered.
In the following, we discuss two of them, called respectively loose and strict initial conditions.
The difference between them is in the time when initial tokens are allowed to contribute to the first firing of transitions.

In \emph{loose initial conditions}, initial tokens can contribute to the firing of transitions at any time, independently from the time interval associated to their initial place.
In particular, no time window constraint is imposed to initial tokens and, therefore, no initial token can become dead in its initial place.
Loose initial conditions are rather permissive, but they can be useful in some specific situations, for example to model the evolution of a system in stationary regime.

\begin{example}[Heat treatment unit, cont.]
If the P-TEG of \Cref{fi:TEG_furnace} works under loose initial conditions, then the first firing of transition $t_2$ can occur at any time $\tau\in\R$.
The first firing of $t_3$ may then not occur before time $\tau+0.5$; for instance, firing at time $\tau + 7$ would not cause token deaths.
\end{example}

\zor{In \emph{strict initial conditions}, instead, initial tokens are assumed to arrive in their initial places at an arbitrary \emph{initial time} $t_0\in\R$.
From this moment, a timer starts keeping track of the time the initial tokens spend in their respective initial places.
In particular, this means that, if the downstream transition of a place with an initial token and temporal upper bound $\tau^+\in\R_{\geq 0}$ does not fire before or at time $t_0+\tau^+$, then the initial token becomes dead.\footnote{\zor{In \cite{zorzenon2023switched}, strict initial conditions have been defined in another way using the concept of \emph{time tags}. It is easy to show that the two definitions are equivalent.}}}

\begin{example}[Heat treatment unit, cont.]\label{ex:heat_treatment_strict}
    \zor{Suppose that the P-TEG in \Cref{fi:TEG_furnace} is subject to strict initial conditions.
    Then, in order to avoid the death of the initial token in $p_{21}$, $t_2$ needs to fire for the first time between time $t_0+2$ and $t_0+3$.
Moreover, observe that the death of the initial token in $p_{33}$ is bound to occur, since the first firing of $t_3$ can only occur after $t_0+6$, because of the constraint associated to $p_{31}$, but should occur before $t_0+4$ to avoid the death of the initial token in $p_{33}$.
    This reasoning shows that the P-TEG, under strict initial conditions, is not consistent.}
\end{example}

\zor{Observe that, since P-TEGs are time-invariant systems, the choice of the value $t_0\in\R$ does not affect their dynamics.}

\subsection{Initial marking transformation}\label{su:marking_transformation}

\zor{We briefly recall that it} is always possible to transform a P-TEG into one in which the number of initial tokens in each place is at most $1$.
The transformation for P-TEGs with loose initial condition is described, \eg in \cite{paek2020analysis}\zor{, and it can be easily extended to the case of strict initial conditions.}

The transformation requires to add new places and transitions, and has the property of not modifying the allowed firing times of the transitions from the original P-TEG.
\zor{In particular, it} increases the number of transitions in the net from $|\transitions|$ to
\begin{equation}\label{eq:additional_transitions_transformation}
    |\transitions| + \sum_{p\in\places} \max\{ 0,\ \marking(p)-1\},
\end{equation}
where $\transitions$ and $\places$ are the set of transitions and places of the original P-TEG.

From now on, we will consider only P-TEGs with at most one \zor{initial} token per place, since any P-TEG can be transformed into a new one with this property.

\subsection{Dynamics as systems of inequalities}

We can now formulate the (nondeterministic) dynamics of a P-TEG $(\places,\transitions,\arcs,\marking,\iota)$ with $|\transitions| = n$ transitions and at most one initial token per place as the precedence constraints presented in the following.
Let us define matrices $A^0,A^1\in\Rmax^{n\times n}$ and $B^0,B^1\in\Rmin^{n\times n}$ such that, if there exists a place $p$ with initial marking $\mu\in\{0,1\}$, upstream transition $t_j$ and downstream transition $t_i$, then $A^\mu_{ij} = \tau_p^-$ and $B^\mu_{ij} = \tau_p^+$, otherwise $A^\mu_{ij} = -\infty$ and $B^\mu_{ij}=+\infty$.
Let $x_i(k)\in\R$ denote the time when transition $t_i\in\transitions$ fires for the $k$-th time, where $i\in\dint{1,n}$ and $k\in\N$.
Since the $(k+1)$-st firing of any transition $t_i$ cannot occur before the $k$-th one, it is natural to assume that $x_i$ is nondecreasing in $k$, i.e., $x_i(k+1) \geq x_i(k)$ for all $i\in\dint{1,n}$ and $k\in\N$.
The dynamics of a P-TEG under loose initial conditions can be described by the following system of infinitely many inequalities in infinitely many variables $x_i(k)$: for all $i,j\in\dint{1,n}$, $\mu\in\{0,1\}$, $k\in\N$,
\begin{equation}\label{eq:dynamics_PTEGs}
	\left\{
	\begin{array}{rcl}
	A^\mu_{ij} + x_j(k) \leq & x_i(k+\mu) & \leq B^\mu_{ij} + x_j(k),\\
	x_i(k) \leq & x_i(k+1). & 
	\end{array}
	\right.
\end{equation}
The meaning of the inequalities in the first line of~\eqref{eq:dynamics_PTEGs} is that, in order to satisfy the constraints imposed by the time interval $[\tau_p^-,\tau_p^+]$ associated to place $p$ with $\marking(p) = \mu$ initial tokens, the downstream transition $t_i$ of $p$ needs to fire for the $(k+\mu)$-th time at least $A_{ij}^\mu = \tau_p^-$ time units and at most $B_{ij}^\mu = \tau_p^+$ time units after the $k$-th firing of the upstream transition $t_j$ of $p$.
The second line of~\eqref{eq:dynamics_PTEGs} simply imposes the nondecreasingness condition on $x_i$.

In the case of strict initial conditions, we need to add inequalities to limit the first firing of transitions with upstream places containing initial tokens.
\zor{In particular, in addition to \eqref{eq:dynamics_PTEGs}, a P-TEG with strict initial conditions must satisfy}: for all $i,j\in\dint{1,n}$,
\begin{equation}\label{eq:dynamics_PTEGs_strict}
	\left\{
	\begin{array}{rcl}
    \zor{A^1_{ij}} + t_0 \leq & x_i(1) & \leq \zor{B^1_{ij}} + t_0,\\
	t_0 \leq & x_i(1). & 
	\end{array}
	\right.
\end{equation}

Note that matrices $A^0,A^1,B^0,B^1$ uniquely define a P-TEG with at most one initial token per place.
For this reason, they are called \emph{characteristic matrices} of the associated P-TEG.
We can now give a more formal definition of consistency: a P-TEG with loose (resp., strict) initial conditions is consistent if there exists an infinite trajectory $\{x_i(k)\in\R\mid i\in\dint{1,n},\ k\in\N\}$ that satisfies~\eqref{eq:dynamics_PTEGs} (resp., and \eqref{eq:dynamics_PTEGs_strict}).
Such a trajectory is then called consistent for the P-TEG, as it corresponds to an evolution of the marking in the P-TEG for which no token death occurs.

\subsection{Consistency with loose initial conditions}

In this section, we consider the following decision problem.

\noindent
P-TEG CONSISTENCY LOOSE\\
\textbf{Instance:} Matrices $A^0,A^1\in\Rmax^{n\times n}$, $B^0,B^1\in\Rmin^{n\times n}$.\\
\textbf{Question:} Is the P-TEG under loose initial conditions with characteristic matrices $A^0$, $A^1$, $B^0$, $B^1$ consistent?

To reduce the number of matrices involved, it is worth stating~\eqref{eq:dynamics_PTEGs} in the following, equivalent form
\begin{equation}\label{eq:simpler_dynamics}
	\begin{array}{l}
		\forall k\in\N,\\
		\forall i,j\in\dint{1,n},\\
	\end{array}
	\quad
	\left\{
	\begin{array}{rl}
		x_i(k) \geq & (L)_{ij} + x_j(k+1),\\
		x_i(k) \geq & (C)_{ij} + x_j(k),\\
		x_i(k+1) \geq & (R)_{ij} + x_j(k),
	\end{array}
	\right.
\end{equation}
where we used matrices $L,C,R\in\Rmax^{n\times n}$ defined by 
\[
 (L)_{ij} = -B^1_{ji}\quad   (C)_{ij} = \max\{A^0_{ij},\,-B^0_{ji}\},\quad 
	(R)_{ij} = 
    \begin{dcases}
        A^1_{ij} & \mbox{if } i\neq j,\\
        \max\{0,\, A^1_{ii}\} & \mbox{if }i = j,
    \end{dcases}.
\]
The equivalence between~\eqref{eq:dynamics_PTEGs} and~\eqref{eq:simpler_dynamics} can be easily verified using the fact that, for any $x,y,a,b\in\R$, the inequalities $x \geq a + y$ and $x\geq b + y$ hold if and only if $x \geq \max\{a + y,b + y\} = \max\{a,b\}+y$.
Even more compactly, we can write,
using the max-plus algebra,
\begin{equation}\label{eq:simple_dynamics_PTEGs}
	\forall k\in\N,
	\quad
	\left\{
	\begin{array}{rl}
		x(k) \geq & L \otimes x(k+1),\\
		x(k) \geq & C \otimes x(k),\\
		x(k+1) \geq & R \otimes x(k).
	\end{array}
	\right.
\end{equation}
Define $x_\N = [x(1)^\top\ x(2)^\top \cdots]^\top\in\R^\N$ and $M_{\N}\in\Rmax^{\N\times \N}$ as in~\eqref{eq:M_N}.
Then, \eqref{eq:simple_dynamics_PTEGs} can be rewritten as the precedence constraints $x_\N\geq M_\N\otimes x_\N$.
Therefore, the P-TEG corresponding to matrices $L,C,R$ is consistent if and only if $x_\N\geq M_\N\otimes x_\N$ admits a real solution $x_\N\in\R^\N$.
From \Cref{th:gallai}, the existence of a real solution is equivalent to $\graph(M_\N)$ not containing $\infty$-weight paths.
Recalling that $\graph(M_\N)$ is an $\N$-periodic graph, \Cref{th:polynomial_complexity} implies that it is possible to solve the problem P-TEG CONSISTENCY LOOSE in strongly polynomial time $O(n^5)$.

\begin{remark}
    According to the definition of P-TEG CONSISTENCY LOOSE, the input size depends on the number of transitions in the considered P-TEG after the marking transformation detailed in \Cref{su:marking_transformation}.
    Observe that the number of transitions added in the P-TEG by the transformation is, in the worst case, linear with respect to the number of initial tokens per place.
    Therefore, our discussion shows that consistency can be verified, in a P-TEG with arbitrary initial marking, in pseudo-polynomial time $O(\tilde{n}^5)$, where  $\tilde{n}$ is the number of transitions in the transformed P-TEG, evaluated in \eqref{eq:additional_transitions_transformation}.
\end{remark}

\begin{example}
\begin{figure}
	\centering
	\begin{tikzpicture}[node distance=.5cm and 1.5cm,>=stealth',bend angle=30]
\footnotesize

\node[transition,label=above:{$t_1$}] (t1) {};
\node[place,right=of t1,label=above:{$[0,\infty)$}] (p12) {};
\node[transition,right=of p12,label=above:{$t_2$}] (t2) {};
\node[place,tokens=1,left= 1cm of t1,label=above:{$[\beta,-\alpha]$}] (p11) {};
\node[place,tokens=1,right= 1cm of t2,label=above:{$[2,3]$}] (p22) {};

\draw (t1) edge[->] (p12);
\draw (p12) edge[->] (t2);
\draw (t1) edge[bend right,->] (p11);
\draw (p11) edge[bend right,->] (t1);
\draw (t2) edge[bend right,->] (p22);
\draw (p22) edge[bend right,->] (t2);

\end{tikzpicture}
	\caption{Example of P-TEG.}
	\label{fi:P-TEG_example}
\end{figure}
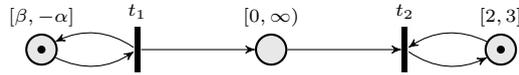

Consider the P-TEG represented in Figure~\ref{fi:P-TEG_example}.
Its characteristic matrices are
	\[A^0 = \begin{bmatrix}
	-\infty & -\infty \\
	0 & -\infty
\end{bmatrix},\quad
	A^1 = \begin{bmatrix}
	\beta & -\infty \\
	-\infty & 2
\end{bmatrix},\quad
	B^0 = \begin{bmatrix}\infty & \infty\\\infty & \infty\end{bmatrix},\quad
	B^1 = \begin{bmatrix}
	-\alpha & \infty \\
	\infty & 3
\end{bmatrix},\] 
where $\alpha\leq 0$, $\beta\geq 0$ are some parameters.
By applying the formulas for $L,C,R$, we get 
	\[L=-B^{1\top}=\begin{bmatrix}
	\alpha & -\infty\\-\infty & -3
	\end{bmatrix},\quad C = A^0\oplus (-B^{0\top})=
	\begin{bmatrix}
		-\infty & -\infty \\ 0 & -\infty
    \end{bmatrix},
	\]\[R = A^1\oplus E = \begin{bmatrix}
	\beta & -\infty\\-\infty & 2
\end{bmatrix}.\]
Note that these matrices coincide with the ones from \Cref{ex:simple_graph}.
Therefore, from the discussion of Example~\ref{ex:simple_graph_3}, we can conclude that the P-TEGs obtained by setting $\alpha=-1,\beta=2$ and $\alpha=-5,\beta=4$ are not consistent, and the one corresponding to values $\alpha=-1,\beta=1$ is consistent.
It can indeed be verified (\eg using the method discussed in the proof of \Cref{th:gallai}) that, for the latter choice of parameters, the trajectory
\[
    x(1) = \begin{bmatrix}
    0\\0
\end{bmatrix},\quad(\forall k\in\N)\
x(k+1) = \begin{bmatrix}
    1\\2
\end{bmatrix} + x(k)
\]
satisfies~\eqref{eq:simpler_dynamics}.
\end{example}

\zor{
    \begin{example}[Heat treatment unit, cont.]\label{ex:heat_treatment_loose}
    Consider the P-TEG \Cref{fi:TEG_furnace} under loose initial conditions.
    To verify whether it is consistent, we first obtain matrices $L,C,R$:
	\[
	    L = \begin{bmatrix}
            \cdot & -3 & \cdot\\
            \cdot & \cdot & \cdot\\
            \cdot & \cdot & -4
\end{bmatrix},\quad
        C = \begin{bmatrix}
            \cdot & 0 & \cdot\\
            \cdot & \cdot & \cdot\\
            \cdot & 0.5 & \cdot
\end{bmatrix},\quad
        R = \begin{bmatrix}
            0 & \cdot & \cdot\\
            2 & 0 & 0.5\\
            6 & \cdot & 0
\end{bmatrix} .%
\]
By applying the algorithm from \cref{se:N_periodic_graphs}, we can observe that the P-TEG is consistent, since the $\N$-periodic graph corresponding to the P-TEG (shown in \Cref{fi:graph_furnace_loose}) does not contain $\infty$-weight paths.
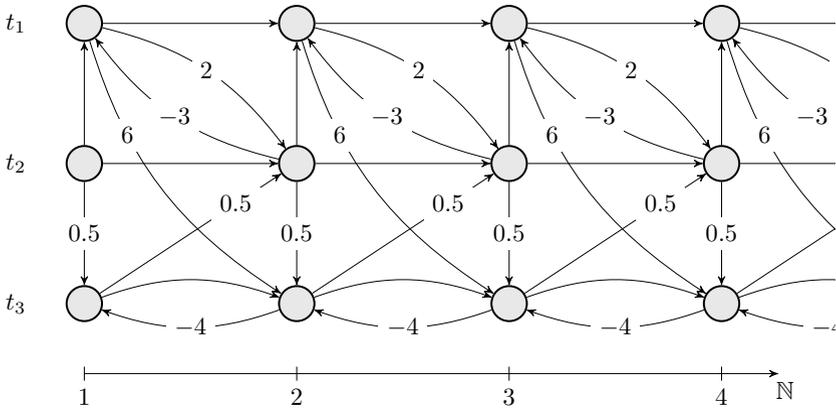
\begin{figure}[h]
    \centering
    \resizebox{\textwidth}{!}{
    \begin{tikzpicture}[node distance=2cm and 1cm,place/.append style={minimum size=.5cm,inner sep=0pt},on grid]

\foreach [count=\i, evaluate=\i as \zo using int(\i-3)] \z in {1,...,10}
{
\ifnum\z<3
\else
    {
    \coordinate (n1\z) at (3*\z,0) {};
    \coordinate (n2\z) at (3*\z,-2) {};
    \coordinate (n3\z) at (3*\z,-4) {};
    }
\fi
}
\clip ($(n33)!.5!(n34)-(.3,1.6)$) rectangle ($(n17)!.5!(n18)+(.1,.5)$);

\foreach [count=\i, evaluate=\i as \zo using int(\i-3)] \z in {1,...,10}
{
\ifnum\z<3
\else
    {
    \footnotesize
    \node [place] (n1\z) at (n1\z) {};
    \node [place] (n2\z) at (n2\z) {};
    \node [place] (n3\z) at (n3\z) {};
    }
    \ifnum\z>3
    {
        \draw [arc] (n2\z) to (n1\z);
        \draw [arc] (n2\z) to node [fill=white] {$0.5$} (n3\z);
    }
    \fi
\fi
}

\foreach[count=\i, evaluate=\i as \zz using int(\i+3)] \z in {3,...,9}
{
    \ifnum\z>3
        {
        \draw [arc] (n1\z) to (n1\zz);
        \draw [arc] (n2\z) to (n2\zz);
        \draw [arc] (n3\z) to [bend left=20] (n3\zz);

        \draw [arc] (n1\z) to [bend left=20] node [fill=white] {$2$} (n2\zz);
        \draw [arc] (n1\z) to [bend right=20] node [pos=.3,fill=white] {$6$} (n3\zz);
        \draw [arc] (n3\z) to node [pos=.75,fill=white] {$0.5$} (n2\zz);
        }
    \fi
}

\foreach[count=\i, evaluate=\i as \zz using int(\i+2)] \z in {4,...,10}
{
    \ifnum\z>4
        {
            \draw [arc] (n3\z) to [bend left=20] node [fill=white] {$-4$} (n3\zz);
            \draw [arc] (n2\z) to [bend left=20] node [fill=white] {$-3$} (n1\zz);
        }
    \fi
}

\draw [-stealth'] ($(n34)+(0,-1cm)$) to node[pos=1.01,below] {$\N$} ($(n37)+(.8,-1cm)$);
\draw ($(n34)+(0,-1.1cm)$) to node[pos=0,below] {$1$} ($(n34)+(0,-.9cm)$);
\draw ($(n35)+(0,-1.1cm)$) to node[pos=0,below] {$2$} ($(n35)+(0,-.9cm)$);
\draw ($(n36)+(0,-1.1cm)$) to node[pos=0,below] {$3$} ($(n36)+(0,-.9cm)$);
\draw ($(n37)+(0,-1.1cm)$) to node[pos=0,below] {$4$} ($(n37)+(0,-.9cm)$);

\node [left=.7cm of n14,anchor=east] {$t_1$};
\node [left=.7cm of n24,anchor=east] {$t_2$};
\node [left=.7cm of n34,anchor=east] {$t_3$};

\end{tikzpicture}
    }
    \caption{$\N$-periodic graph corresponding to the P-TEG in \Cref{fi:TEG_furnace} with loose initial conditions. When not otherwise specified, arcs have weight $0$.}\label{fi:graph_furnace_loose}
\end{figure}
\end{example}
}

\subsection{Consistency with strict initial conditions}

Let us consider now the following decision problem.

\noindent
P-TEG CONSISTENCY STRICT\\
\textbf{Instance:} Matrices $A^0,A^1\in\Rmax^{n\times n}$, $B^0,B^1\in\Rmin^{n\times n}$.\\
\textbf{Question:} Is the P-TEG under strict initial conditions with characteristic matrices $A^0$, $A^1$, $B^0$, $B^1$ consistent?

Proceeding as in the previous section, we determine that any consistent trajectory must satisfy \eqref{eq:simple_dynamics_PTEGs}.
In addition, in order to consider strict initial conditions we now restate \eqref{eq:dynamics_PTEGs_strict} as precedence constraints.
\zor{Define $C_\wT$ as the $n\times n$ matrix such that, for all $i,j\in\dint{1,n}$, $(C_\wT)_{ij} = 0$.}
Consider the inequalities
\begin{equation}\label{eq:simple_dynamics_PTEGs_strict}
	\left\{
	\begin{array}{rl}
        x(0) \geq & \zor{L} \otimes x(1),\\
        x(0) \geq & \zor{C_\wT} \otimes x(0),\\
            x(1) \geq & \zor{R} \otimes x(0).
	\end{array}
	\right.
\end{equation}
Since $(\zor{C_\wT})_{ij} = 0$ for all $i,j$, the second inequality imposes that $x_{i}(0) = x_j(0)$ for all $i,j$.
Given the arbitrariness of the initial time, we can interpret $x_i(0)$ as $t_0$ for all $i\in\dint{1,n}$.
Then, the first and third inequalities of \eqref{eq:simple_dynamics_PTEGs_strict} coincide with \eqref{eq:dynamics_PTEGs_strict}. 

Define $x_{\textup{U}} = [x(0)^\top\ x(1)^\top\ x(2)^\top \cdots]^\top\in\R^\N$ and $M_{\textup{U}}\in\Rmax^{\N\times \N}$ as
\[
    M_\textup{U} = 
    \begin{bmatrix}
        C_\wT & L & \cdot & \cdot & \cdots\\
        R & C & L & \cdot & \cdots\\
        \cdot & R & C& L & \cdots\\
        \cdot & \cdot & R & C& \cdots\\
        \vdots & \vdots & \vdots & \vdots & \ddots
    \end{bmatrix}\in\Rmax^{\N\times \N}.
\]
Then, the precedence constraints \eqref{eq:simple_dynamics_PTEGs} and \eqref{eq:simple_dynamics_PTEGs_strict} can be rewritten as $x_{\textup{U}}\geq M_{\textup{U}}\otimes x_{\textup{U}}$.
Observe that matrix \zor{$M_{\textup{U}}$} can be padded by adding infinitely many rows of $-\infty$'s on the top and infinitely many columns of $-\infty$'s on the left, so that the resulting matrix has the form of \eqref{eq:M_U}.
In particular, the corresponding precedence graph is an ultimately periodic graph of the form $\graph_\textup{U}(\graph_\wN,C_\wT,\graph_\wP)$, where $\graph_\wN = \graph(\mathcal{E},\mathcal{E},\mathcal{E})$,
and \zor{$\graph_\wP = \graph(L,C,R)$}. %
Using \Cref{th:ultimately_periodic}, we conclude that the problem P-TEG CONSISTENCY STRICT can be solved in strongly polynomial time $O(n^5)$.

\zor{
For the sake of convenience, the algorithm that verifies consistency in P-TEGs under strict initial conditions is summarized in \Cref{al:consistency_strict}.
\Cref{al:consistency_strict} can be adjusted so that it verifies consistency in P-TEGs under loose initial conditions by removing lines 2, 11, 13, and 14.

\begin{algorithm2e}
  \small
  \DontPrintSemicolon
  \caption{Verify consistency under strict initial conditions}
  \label{al:consistency_strict}
  \SetKwInOut{Input}{Input}
  \SetKwInOut{Output}{Output}
  \Indmm
  \Input{$A^0,A^1,B^0,B^1\in\Rmax^{n\times n}$}
  
  \Output{{\tt true} iff the P-TEG characterized by $A^0,A^1,B^0,B^1$ is consistent under strict initial conditions}
  \Indpp
  \BlankLine
  $L = -B^{1\top},\ C = A^0\oplus (-B^{0\top}),\ R = A^1\oplus E$\;
  Define $C_\wT$ as the $n\times n$ matrix with only $0$'s\;
  \If{$\graph(C)\not\in\nonegset$}{
      \Return {\tt false}
  }
  $\Pi_\wP(0) = C^*$\;
  \For{$h = 0$ to $n^2$}{
      \If{$\graph(L\Pi_\wP(h) R \oplus C)\not\in\nonegset$}{
          \Return {\tt false}
      }
      $\Pi_\wP(h+1) = (L\Pi_\wP(h) R \oplus C)^*$\;
      \If{$\Pi_\wP(h)=\Pi_\wP(h+1)$}{
            \If{$\graph(C_\wT\oplus L\Pi_\wP(h+1)R)\in\nonegset$}{
                \Return {\tt true}
            }\Else{
                \Return {\tt false}
            }
      }
  }
  \Return {\tt false}
\end{algorithm2e}
}

\begin{example}[Heat treatment unit, cont.]
    Consider the P-TEG \Cref{fi:TEG_furnace} \zor{under} strict initial conditions.
    \zor{As seen in \Cref{ex:heat_treatment_strict}, this P-TEG is not consistent; we can now formally verify this.}
    \zor{Matrices $L,C,R$ for this P-TEG have already been computed in \Cref{ex:heat_treatment_loose}; additionally, we define $C_\wT\in\R^{4\times 4}$ such that $(C_\wT)_{ij}=0$ for all $i,j$.}
The ultimately periodic graph associated to these matrices is represented in \Cref{fi:graph_furnace}.
As highlighted, the graph contains an elementary circuit with positive weight (equal to $2$); therefore, no consistent trajectory exists for the heat treatment unit.
\begin{figure}[h]
    \centering
    \resizebox{\textwidth}{!}{
    \begin{tikzpicture}[node distance=2cm and 1cm,place/.append style={minimum size=.5cm,inner sep=0pt},on grid]

\foreach [count=\i, evaluate=\i as \zo using int(\i-3)] \z in {1,...,10}
{
\ifnum\z<3
\else
    {
    \coordinate (n1\z) at (3*\z,0) {};
    \coordinate (n2\z) at (3*\z,-2) {};
    \coordinate (n3\z) at (3*\z,-4) {};
    }
\fi
}
\clip ($(n33)!.5!(n34)-(.3,1.6)$) rectangle ($(n17)!.5!(n18)+(.1,.5)$);

\foreach [count=\i, evaluate=\i as \zo using int(\i-3)] \z in {1,...,10}
{
\ifnum\z<3
\else
    {
    \footnotesize
    \node [place] (n1\z) at (n1\z) {};
    \node [place] (n2\z) at (n2\z) {};
    \node [place] (n3\z) at (n3\z) {};
    }
    \ifnum\z=4
        {
            \draw [-stealth'] (n1\z) to [bend left=10] (n2\z);
            \draw [-stealth'] (n2\z) to [bend left=10] (n3\z);
            \draw [-stealth'] (n2\z) to [bend left=10] (n1\z);
            \draw [-stealth'] (n3\z) to [bend left=10] (n2\z);
            \draw [-stealth'] (n1\z) to [bend right=20] (n3\z);
            \draw [-stealth'] (n3\z) to [bend left=30] (n1\z);
        }
    \fi
    \ifnum\z>4
    {
        \draw [arc] (n2\z) to (n1\z);
        \draw [arc] (n2\z) to node [fill=white] {$0.5$} (n3\z);
    }
    \fi
\fi
}

\foreach[count=\i, evaluate=\i as \zz using int(\i+3)] \z in {3,...,9}
{
    \ifnum\z>3
        {
        \draw [arc] (n1\z) to (n1\zz);
        \draw [arc] (n2\z) to (n2\zz);
        \draw [arc] (n3\z) to [bend left=20] (n3\zz);

        \draw [arc] (n1\z) to [bend left=20] node [fill=white] {$2$} (n2\zz);
        \draw [arc] (n1\z) to [bend right=20] node [pos=.3,fill=white] {$6$} (n3\zz);
        \draw [arc] (n3\z) to node [pos=.75,fill=white] {$0.5$} (n2\zz);
        }
    \fi
}

\foreach[count=\i, evaluate=\i as \zz using int(\i+2)] \z in {4,...,10}
{
    \ifnum\z>4
        {
            \draw [arc] (n3\z) to [bend left=20] node [fill=white] {$-4$} (n3\zz);
            \draw [arc] (n2\z) to [bend left=20] node [fill=white] {$-3$} (n1\zz);
        }
    \fi
}

\draw [-stealth'] ($(n34)+(0,-1cm)$) to node[pos=1.01,below] {$\No$} ($(n37)+(.8,-1cm)$);
\draw ($(n34)+(0,-1.1cm)$) to node[pos=0,below] {$0$} ($(n34)+(0,-.9cm)$);
\draw ($(n35)+(0,-1.1cm)$) to node[pos=0,below] {$1$} ($(n35)+(0,-.9cm)$);
\draw ($(n36)+(0,-1.1cm)$) to node[pos=0,below] {$2$} ($(n36)+(0,-.9cm)$);
\draw ($(n37)+(0,-1.1cm)$) to node[pos=0,below] {$3$} ($(n37)+(0,-.9cm)$);

\node [left=.7cm of n14,anchor=east] {$t_1$};
\node [left=.7cm of n24,anchor=east] {$t_2$};
\node [left=.7cm of n34,anchor=east] {$t_3$};

\draw [glow=myred,draw=none] (n14) to [bend right=20] (n35);
\draw [glow=myred,draw=none] (n35) to [bend left=20] (n34);
\draw [glow=myred,draw=none] (n34) to [bend left=30] (n14);

\end{tikzpicture}
    }
    \caption{Ultimately periodic graph corresponding to the P-TEG in \Cref{fi:TEG_furnace} with strict initial conditions. When not otherwise specified, arcs have weight $0$.}\label{fi:graph_furnace}
\end{figure}
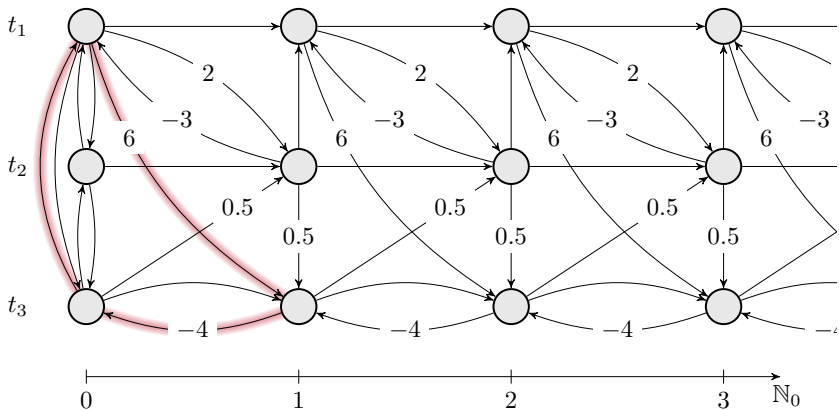
\end{example}

\zor{\begin{example}\label{ex:pteg_nonperiodic}
\Cref{fi:pteg_nonperiodic} shows an example of consistent P-TEG under strict initial conditions that does not admit any periodic trajectory in the sense of \cite{lee2005extended,munier2011graph}, \ie such that, for all $i\in\dint{1,4}$ there is a $\lambda_i\in\R$ such that, for all $k\in\N$, $x_i(k+1) = \lambda_i + x_i(k)$.
Indeed, it is easy to see that, assuming the initial time $t_0=0$, the only consistent trajectory has the form
\[
    x(1) = \begin{bmatrix}
        0\\1\\1\\3
    \end{bmatrix},\ 
    x(2) = \begin{bmatrix}
        3\\4\\2\\4
    \end{bmatrix},\ 
    \forall k\geq 1,\ x(k+2) = 4\otimes x(k).
\]
The consistency of the P-TEG can also be verified by following the algorithm described in \Cref{su:longest_path_ultimately_periodic}.
The matrices $L,C,R$ and $C_\wT$ associated with the P-TEG are
\[
    L = \begin{bmatrix}
        \cdot &\cdot &\cdot &\cdot \\
        \cdot &\cdot &-1 &\cdot \\
        \cdot &\cdot &\cdot &\cdot \\
        0 &\cdot &\cdot &\cdot \\
    \end{bmatrix},\quad
    C = \begin{bmatrix}
        \cdot &-1 &\cdot &\cdot \\
        1 &\cdot &\cdot &\cdot \\
        \cdot &\cdot &\cdot &-2 \\
        \cdot &\cdot &2 &\cdot \\
    \end{bmatrix},\quad
    R = \begin{bmatrix}
        0 &\cdot &\cdot &0 \\
        \cdot &0 &\cdot &\cdot \\
        \cdot &1 &0 &\cdot \\
        \cdot &\cdot &\cdot &0 \\
    \end{bmatrix}
\]
and $(C_\wT)_{ij} = 0$ for all $i,j$.
Consistency is then implied by the following facts: the sequence
\[
    \forall h\in\No\quad \Pi_\wP(0) = C^+,\quad 
    \Pi_\wP(h+1) = (L\Pi_\wP(h)^*R\oplus C)^+
\]
converges after $h=1$ step, and $\graph(C_\wT\oplus L\Pi_\wP(1)^*R) = \graph(
            C_\wT
    )\in\nonegset$.
\begin{figure}[t]
    \centering
    \begin{tikzpicture}[node distance=.2cm and 1cm,>=stealth',bend angle=30]
\footnotesize

\node[transition,label=above:{$t_1$}] (t1) {};
\node[place,above right=of t1,label=above:{$[1,1]$}] (p12) {};
\node[transition,right=of p12,label=above:{$t_2$}] (t2) {};
\node[place,tokens=1,right=of t2,label=above:{$[1,1]$}] (p23) {};
\node[transition,below right=of p23,label=above:{$t_3$}] (t3) {};
\node[place,below left=of t3,label=above:{$[2,2]$}] (p34) {};
\node[transition,left=of p34,label=above:{$t_4$}] (t4) {};
\node[place,tokens=1,left=of t4,label=above:{$[0,0]$}] (p41) {};

\draw [->] (t1.90-15) to [bend left] (p12);
\draw [->] (p12) to (t2);
\draw [->] (t2) to (p23);
\draw [->] (p23) to [bend left] (t3.90+15);

\draw [->] (t3.-90-15) to [bend left] (p34);
\draw [->] (p34) to (t4);
\draw [->] (t4) to (p41);
\draw [->] (p41) to [bend left] (t1.-90+15);

\end{tikzpicture}
    \caption{Example of P-TEG that, under strict initial conditions, admits only non-periodic trajectories.}\label{fi:pteg_nonperiodic}
\end{figure}
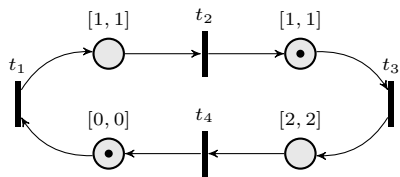
\end{example}}

\section{Conclusions}\label{se:conclusions}

In \Cref{th:gallai}, we showed that a system of infinitely many precedence constraints admits solutions if and only if the corresponding precedence graph does not contain an $\infty$-weight path.
This result was then applied to $\N$-periodic and ultimately periodic graphs.
We remark that there are much larger classes of graphs with interesting applications where \Cref{th:gallai} could be used.
Examples are linear graphs \cite{munier1996basic}, equivalent to weighted timed event graphs and synchronous data-flows, and multi-dimensional periodic graphs (such as the ones studied in \cite{hoefting1995minimum}), which model very large scale integration circuits.
Moreover, note that the fields of application of ultimately periodic graphs is not limited to P-TEGs with strict initial conditions.
For example, ultimately periodic graphs can be used for the analysis of switched max-plus linear-dual inequalities under ultimately periodic schedules \cite{zorzenon2023switched}.

To conclude, we point out an interesting connection between $\N$-periodic graphs (and therefore P-TEGs) and vector addition systems with states (VASSes).
It was first noticed by Gaubert \cite{gaubert2023personal}, and it inspired the technical lemmas in \Cref{se:N_periodic_graphs}.
An $\X$-VASS, where the set $\X\subseteq \Z^d$ is typically either $\Z^d$ or $\N^d$ for some $d\in\N$, is a dynamical systems represented by a finite directed graph $V = (\nodes,\arcs)$, in which the set of arcs $\arcs$ is a subset of $\nodes\times \Z^d \times \nodes$.
Given a configuration $(q,x)\in\nodes\times \X$ of an $\X$-VASS, we say that configuration $(q',x')$ is reachable in one step if $(q,x'-x,q')\in\arcs$ and $x'\in\X$.
In $\N^d$-VASSes the latter condition imposes components of vector $x'$, called counters, to be positive.
Reachability is then extended to multiple steps in an obvious way.
$\N^d$-VASSes are equivalent, up to a polynomial-time transformation, to untimed Petri nets.
The restriction on the positivity of counters in $\N^d$-VASSes translates into the impossibility for places in Petri nets to have a negative number of tokens.

There exists a vast theory on $\Z^d$- and, especially, $\N^d$-VASSes.
The central algorithmic problem for these systems is the reachability problem, asking whether it is possible to reach a configuration from an initial one.
While it is folklore that this problem is NP-complete in $\Z^d$-VASSes \cite{haase2014integer}, only recently it has been proven that the reachability problem in $\N^d$-VASSes is Ackermann-complete \cite{czerwinski2022reachability}.
The restriction to positive counters thus makes $\N^d$-VASSes much harder to deal with.

Let $V$ be a $\Z^2$-VASS in which $\arcs = \nodes\times (\Z\times \{-1,0,+1\})\times \nodes$, \ie each arc from node $q\in\nodes$ to $q'\in\nodes$ is associated to a pair $(w\ s)\in\Z\times\{-1,0,+1\}$.
It is possible to construct a $\Z$-periodic graph $\graph_\Z$ with integer arcs weight, such that the question "is configuration $(q',(w'\ s'))$ in $V$ reachable from $(q,(w\ s))$?" translates into "is there a path in $\graph_\Z$ from node $(q,s)$ to node $(q',s')$ of weight $w'-w$?".
Therefore, reachability problems (among others) can be equivalently stated in $\Z^2$-VASSes and $\Z$-periodic graphs.
In the same way, one can show that $\N$-periodic graphs are equivalent to $\X$-VASSes where $\X = \Z\times\N$, \ie where only one of the two counters is restricted to the set of positive integers.
This class of VASSes, which represents a sort of hybrid version between $\N^2$- and $\Z^2$-VASSes, has not been studied in the literature.

Despite the equivalence between these models, the research on infinite precedence graphs seems to have evolved independently from that on VASSes.
In particular, the VASS equivalent of the problem discussed in \Cref{se:N_periodic_graphs,se:ultimately_periodic_graphs} -- which asks, "given a configuration $(q,(w\ s))$, a state $q'$, and number $s'$, is the supremal value of $w'$, such that configuration $(q',(w',\ s'))$ is reachable from $(q,(w\ s))$, finite?" -- has not been considered.
Ultimately periodic graphs are another model with no corresponding in the VASS research.
Given the highlighted connection between VASSes and infinite precedence graphs, we expect many insights from one field to be transferable to the other.

\appendix

\section{Proof of \Cref{le:length_supremal_path_N_periodic_graph}}\label{se:prooffirstlemma}

\firstlemma*
\begin{proof}
Let $\hat{p}_{(1)}$ be a path in $\graph_\N$ from $(i_{(1)},1)$ to $(j_{(1)},1)$ with maximal weight and minimal right-shift, where $i_{(1)}=i$ and $j_{(1)}=j$.
Note that, because this weight is finite and there are finitely many different weights of arcs in $\graph_\N$, such \zor{a} path must exist, in the sense that the supremal weight of all paths from $(i_{(1)},1)$ to $(j_{(1)},1)$ must be attained by a (finite) path.
\zor{Since $\myshift(\myup(\hat{p}_{(1)}))=1$,} clearly $\hat{p}_{(1)}$ has zero left-shift.

If $\myrshift(\hat{p}_{(1)}) > 1$, this path can be factored into $\hat{p}_{(1)} = \hat{s}_{(1)}\hat{p}_{(2)}\hat{t}_{(1)}$\footnote{There may be several ways to factorize $\hat{p}_{(1)}$ in this way. The following arguments can be applied to all of them.}, where $\hat{p}_{(2)}$ is a nonempty path from $(i_{(2)},2)$ to $(j_{(2)},2)$ with zero left-shift and right-shift $\myrshift(\hat{p}_{(2)}) = \myrshift(\hat{p}_{(1)})-1$, $\hat{s}_{(1)}$ is a path from $(i_{(1)},1)$ to $(i_{(2)},2)$, and $\hat{t}_{(1)}$ is a path from $(j_{(2)},2)$ to $(j_{(1)},1)$. 

We now prove by contradiction that $(i_{(1)},j_{(1)})\neq (i_{(2)},j_{(2)})$.
If $(i_{(1)},j_{(1)}) = (i_{(2)},j_{(2)})$, then $\hat{s}_{(1)}$ and $\hat{t}_{(1)}$ are two pseudo-circuits, and their corresponding paths $s_{(1)}$ and $t_{(1)}$ in the static graph $\graph$ are circuits.
Let $p_{(1)}$ (resp., $p_{(2)}$) be the path in $\graph$ corresponding to $\hat{p}_{(1)}$ (resp., $\hat{p}_{(2)}$), and suppose that $\myweight(s_{(1)})+\myweight(t_{(1)})>0$. 
Then, we can construct a sequence of paths $s_{(1)}^{x}p_{(2)}t_{(1)}^{x}$ in $\graph$ from node $i_{(1)}$ to node $j_{(1)}$ with weight increasing in $x\in\N$.
Observe that all such paths have zero left-shift; thus, they induce paths in $\graph_\N$ from $(i_{(1)},1)$ to $(j_{(1)},1)$ with increasing weight.
But this implies that there exists an $\infty$-weight path in $\graph_\N$ from $(i_{(1)},1)$ to $(j_{(1)},1)$, which is in contradiction with our hypotheses.
    On the other hand, if $(i_{(1)},j_{(1)}) = (i_{(2)},j_{(2)})$ and $\myweight(s_{(1)})+\myweight(t_{(1)})\leq 0$, then $\myweight(p_{(2)})\geq \myweight(p_{(1)})$.
As $p_{(2)}$ has zero left-shift, it induces a path in $\graph_{\N}$ from $(i_{(1)},1)$ to $(j_{(1)},1)$ with larger or equal weight and smaller right-shift compared to $\hat{p}_{(1)}$.
    This, once again, contradicts our hypotheses, as $\hat{p}_{(1)}$ is supposed to be the path with maximal weight and minimal right-shift from $(i_{(1)},1)$ to $(j_{(1)},1)$.
    Therefore, $(i_{(1)},j_{(1)}) \neq (i_{(2)},j_{(2)})$.

    It should be clear, from the definition of $\hat{p}_{(1)}$, that $\hat{p}_{(2)}$ is the path with maximal weight and minimal right-shift among those with zero left-shift from $(i_{(2)},2)$ to $(j_{(2)},2)$.
    Therefore, like $\hat{p}_{(1)}$, if $\myrshift(\hat{p}_{(2)}) = \myrshift(\hat{p}_{(1)}) - 1> 1$, it can be factored into $\hat{p}_{(2)} = \hat{s}_{(2)}\hat{p}_{(3)}\hat{t}_{(2)}$, where $\hat{p}_{(3)}$ is a path from $(i_{(3)},3)$ to $(j_{(3)},3)$ with zero left-shift and right-shift $\myrshift(\hat{p}_{(3)}) = \myrshift(\hat{p}_{(2)})-1$, $\hat{s}_{(2)}$ is a path from $(i_{(2)},2)$ to $(i_{(3)},3)$, and $\hat{t}_{(2)}$ is a path from $(j_{(3)},3)$ to $(j_{(2)},2)$.
    We therefore can repeat the same reasoning as before and show that $(i_{(3)},j_{(3)})\neq (i_{(2)},j_{(2)})$.
    The same reasoning can also be used to show that $(i_{(1)},j_{(1)})\neq (i_{(3)},j_{(3)})$.

Suppose, by means of contradiction, that $\myrshift(\hat{p}_{(1)}) > n^2$.
At the $(n^2+1)$-st application of the above procedure, we have found $n^2+1$ different pairs $(i_{(1)},j_{(1)}),\dots,(i_{(n^2+1)},j_{(n^2+1)})$ of numbers from $\dint{1,n}$.
But this is impossible, since there are exactly $n^2$ different pairs of numbers from $\dint{1,n}$.
Therefore, $\myrshift(\hat{p}_{(1)})\leq n^2$.
\end{proof}

\section{Proof of \Cref{le:formula_Pi}}\label{se:proofsecondlemma}

\secondlemma*
\begin{proof}
For $h = 0$, $\Pi_{ij}(0)$ coincides with the largest weight of paths in $\graph_{\dint{1,1}}$ from node $(j,1)$ to node $(i,1)$ or, equivalently, with the largest weight of paths in $\graph(C)$ from node $j$ to $i$; therefore $\Pi(0) = C^+$.

Observe that $\Pi(h+1)$ is the top-left $n\times n$ block of matrix
\begin{equation*}\label{eq:A_1h_plus}
    (M_{\dint{1,h+2}})^+ = 
    \begin{bNiceArray}{c|ccccc}
    C & L & \cdot & \cdots & \cdot & \cdot\\\hline
    R & C & L & \cdots& \cdot & \cdot\\
    \cdot & R & C& \cdots& \cdot & \cdot\\
    \vdots & \vdots & \vdots & \ddots& \vdots & \vdots\\
    \cdot & \cdot & \cdot & \cdots & C & L\\
    \cdot & \cdot & \cdot & \cdots & R & C
\end{bNiceArray}^+\in\Rbar^{((h+2)n)\times ((h+2)n)}.
\end{equation*}
Recall (\eg from \cite{baccelli1992synchronization}) the following formula for the Kleene plus in block matrices:
\[
    \begin{bmatrix}
        a & b\\c & d
    \end{bmatrix}^+ = 
    \begin{bmatrix}
        (\zor{bd^*c}\oplus a)^+ & \times\\ \times& \times
    \end{bmatrix},
\]
where each "$\times$" indicates an element not of interest for our discussion.
By direct application of the formula to compute the top-left $n\times n$ block of matrix $(M_{\dint{1,h+2}})^+$, we get
\[
    \Pi(h+1) = \left(
    \begin{bmatrix}
        L & \cdot & \cdots & \cdot & \cdot
    \end{bmatrix}\otimes (M_{\dint{1,h+1}})^* \otimes
    \begin{bmatrix}
        R\\\cdot\\\vdots\\\cdot\\\cdot
    \end{bmatrix} \oplus C
    \right)^+ = (L (\Pi(h)\oplus E) R \oplus C)^+.
\]
The proof is concluded by observing that $\Pi(h)\oplus E = \Pi(h)^*$.
\end{proof}

\section{Proof of \Cref{le:inf_weight_paths_N_periodic}}\label{se:proofthirdlemma}

\thirdlemma*
\begin{proof}
We have already seen that, if either $\Pi(n^2+1)\neq \Pi(n^2)$ or $\Pi(n^2)\not\in \Rmax^{n\times n}$, then there are $\infty$-weight paths in $\graph_\N$.
    Therefore, we just need to prove direction "$\Leftarrow$".
    As discussed above, if $\Pi(n^2+1) = \Pi(n^2)\in\Rmax^{n\times n}$, then there are no $\infty$-weight paths from node $(i,1)$ to node $(j,1)$, for all $i,j\in\dint{1,n}$.
We now show, by contrapositive, that this is enough to conclude that $\graph_\N$ does not contain $\infty$-weight paths from any node to any other node.
    In other words, we prove that if $\infty$-weight paths exist, then there is an $\infty$-weight path between nodes with unitary shift. 

Assume that $\graph_\N$ contains $\infty$-weight paths.
    According to \Cref{re:elementary_paths_circuits}, there are two possibilities: either there exists a positive-weight circuit or an infinite sequence of elementary paths with infinite limit-weight (or both).
    Suppose that a positive-weight circuit exists.
This circuit corresponds to a circuit in $\graph$ with positive weight and zero shift which, in turn, induces a circuit with positive weight visiting at least one node of the form $(i,1)$.
    Therefore, there exists a circuit with positive weight starting from node $(i,1)$.
    Then, as discussed above, either $\Pi(n^2)\not\in \Rmax^{n\times n}$ if the circuit has right-shift at most $n^2$ or $\Pi(n^2+1)\neq \Pi(n^2)$ otherwise.

    Suppose now that no positive-weight circuit exists.
    Then, there must be a sequence of elementary paths $p^{(1)},p^{(2)},\dots$ with infinite limit-weight from node $(i,h)$ to node $(j,k)$, for some $i,j\in\dint{1,n}$ and $h,k\in\N$.
    Let $h\leq k$; the proof for the case $h>k$ is analogous and therefore will not be explicitly addressed.
    Note that, a path from the sequence cannot visit nodes with shift $k$ for more than $n$ times, otherwise it would not be elementary.
    Therefore, any path $p^{(\ell)}$ from the sequence can be factored into
    \[
        p^{(\ell)} = w_1^{(\ell)} t_1^{(\ell)} w_2^{(\ell)} t_2^{(\ell)}\cdots w_{n-1}^{(\ell)} t_{n-1}^{(\ell)} w_n^{(\ell)},
    \]
    where $w_1^{(\ell)},\ldots,w_n^{(\ell)}$ are (possibly empty) paths visiting only nodes with shifts at most $k$, and $t_1^{(\ell)},\ldots,t_{n-1}^{(\ell)}$ are paths visiting only nodes with shift at least $k$.
    Moreover, since the number of nodes with shift at most $k$ is finite and each path $w_1^{(\ell)},\ldots,w_n^{(\ell)}$ is elementary, by the infinite pigeonhole principle there must be an infinite subsequence of paths $\tilde{p}^{(1)},\tilde{p}^{(2)},\dots$ with infinite limit-weight and with factorization
    \[
        \tilde{p}^{(\ell)} = w_1 t_1^{(\ell)} w_2 t_2^{(\ell)}\cdots w_{n-1} t_{n-1}^{(\ell)} w_n,
    \]
    where $w_1,\dots,w_n$ are independent on $\ell$.
    Since the weights of $w_1,\dots,w_n$ are constant in $\ell$, then there must exists at least one index $h\in\dint{1,n-1}$ such that the weight of $t^{(\ell)}_h$ has infinite weight for $\ell\rightarrow +\infty$.
    For the same reason, also the source and target nodes of $t^{(\ell)}_h$ are independent on $\ell$; we denote them by $\myup(t_h^{(\ell)})=(i',k)$ and $\mydown(t_h^{(\ell)})=(j',k)$.
    Thus, the sequence $t_h^{(1)},t_h^{(2)},\dots$ forms an $\infty$-weight path from $(i',k)$ to $(j',k)$.
    By definition, each path $t^{(\ell)}_h$ has zero left-shift.
Therefore, it can be translated so that its source node has shift $1$; in this way, we have obtained an $\infty$-weight path from $(i',1)$ to $(j',1)$ in $\graph_\N$.
\end{proof}

\section{Proof of \Cref{pr:infinite_weight_ultimately_periodic}}\label{se:prooffourthlemma}

\fourthlemma*
\begin{proof}
    Both directions of the proof are done by contrapositive.

    "$\Rightarrow$":
The graph $\graph_\N(R_\wN,C_\wN,L_\wN)$ is equivalent (up to renaming its nodes) to the negative periodic part of $\graph_\textup{U}$. 
The graph $\graph_\N(L_\wP,C_\wP,R_\wP)$, instead, coincides with the positive periodic part of $\graph_\textup{U}$.
Clearly, if one of these two graphs contains $\infty$-weight paths, then there are $\infty$-weight paths also in $\graph_{\textup{U}}$.
    From the discussion in \Cref{su:bounds}, we can see that 
    \[
    ((R_\wN \Pi_\wN(n^2)^* L_\wN \oplus C_\wT\oplus L_\wP \Pi_\wP(n^2)^* R_\wP )^+)_{ji}
    \]
    is the supremal weight of all paths starting from node $(i,0)$ and ending in $(j,0)$ with left-shift at least $-n^2$ and right-shift at most $n^2$.
    Therefore, if this value is $+\infty$, there exist\zor{s} an $\infty$-weight path in $\graph_\textup{U}$.

"$\Leftarrow$": we start by showing that, if $\graph_\textup{U}$ contains $\infty$-weight paths but $\graph_\N(R_\wN,C_\wN,L_\wN)$ and $\graph_\N(L_\wP,C_\wP,R_\wP)$ do not contain $\infty$-weight paths, then there must be a positive-weight circuit in $\graph_\textup{U}$ passing through node $(i,0)$ for some $i\in\dint{1,n}$.
Observe that, with these conditions, there cannot be a positive-weight circuit that does not visit any node with zero shift, otherwise either $\graph_\N(R_\wN,C_\wN,L_\wN)$ or $\graph_\N(L_\wP,C_\wP,R_\wP)$ contains $\infty$-weight paths.
Therefore, according to \Cref{re:elementary_paths_circuits}, we want to show that with these hypotheses there cannot be a sequence of elementary paths in $\graph_\textup{U}$ with infinite limit-weight.

    Suppose, by means of contradiction, that such a sequence exists.
Then, infinitely many paths of the sequence must visit at least one node with zero shift, otherwise either $\graph_\N(R_\wN,C_\wN,L_\wN)$ or $\graph_\N(L_\wP,C_\wP,R_\wP)$ contains $\infty$-weight paths.
    On the other hand, no path can visit a node with zero shift for more than $n$ times, otherwise it is not elementary.
    By the infinite pigeonhole principle, we can therefore find a subsequence of elementary paths $p^{(1)},p^{(2)},\dots$ that can be factored as 
    \[
        p^{(k)} = t^{(k)}_1w_1t^{(k)}_2w_2\cdots t^{(k)}_nw_nt^{(k)}_{n+1},
    \]
    where each $t^{(k)}_j$ is a (possibly empty) path visiting only nodes with non-zero shift, and each $w_j$ is a path (independent on $k$) visiting only nodes with shift $s\in\{-1,0,+1\}$.
    Since the sum of the weight of paths $w_1,\dots,w_n$ is independent on $k$, if the weight of $p^{(k)}$ tends to $+\infty$ for $k\rightarrow +\infty$, then there must exist at least one index $j\in\dint{1,n+1}$ such that the weight of $t^{(k)}_j$ tends to $+\infty$ for $k\rightarrow +\infty$.
But this implies that either $\graph_\N(R_\wN,C_\wN,L_\wN)$ or $\graph_\N(L_\wP,C_\wP,R_\wP)$ contains $\infty$-weight paths, which contradicts our hypotheses.
Therefore, there must exist a positive-weight circuit in $\graph_\textup{U}$ visiting at least one node with zero shift.
    
    Note that, using the formulas from \Cref{su:bounds}, the supremal weight of all circuits visiting node $(i,0)$ with left-shift at least $-k$ and right-shift at most $k$ (where $k\in\N$) is
    \[
        ((R_\wN \Pi_\wN(k)^* L_\wN \oplus C_\wT\oplus L_\wP\Pi_\wP(k)^* R_\wP )^+)_{ii}.
    \]
    Therefore, the supremal weight of all circuits visiting node $(i,0)$ (without conditions on their left- and right-shift) is
    \[
        \lim_{k\rightarrow +\infty}
        ((R_\wN \Pi_\wN(k)^* L_\wN \oplus C_\wT\oplus L_\wP\Pi_\wP(k)^* R_\wP )^+)_{ii}.
    \]
We have already seen in \Cref{su:bounds} that if $\graph_\N(L_\wP,C_\wP,R_\wP)$ does not contain $\infty$-weight paths, then $\lim_{k\rightarrow \infty} \Pi_\wP(k)= \Pi_\wP(n^2)$.
A similar reasoning shows that, if $\graph_\N(R_\wN,C_\wN,L_\wN)$ does not contain $\infty$-weight paths, then $\lim_{k\rightarrow \infty} \Pi_\wN(k)= \Pi_\wN(n^2)$.
We conclude that, if $\graph_\textup{U}$ contains $\infty$-weight paths, then either one of the graphs $\graph_\N(L_\wP,C_\wP,R_\wP)$ and $\graph_\N(R_\wN,C_\wN,L_\wN)$ contains $\infty$-weight paths, or there exists an entry in the diagonal of
    \[
        (R_\wN \Pi_\wN(n^2)^* L_\wN \oplus C_\wT\oplus L_\wP\Pi_\wP(n^2)^* R_\wP )^+
    \]
    that is equal to $+\infty$, \ie $\graph(R_\wN \Pi_\wN(n^2)^* L_\wN \oplus C_\wT\oplus L_\wP\Pi_\wP(n^2)^* R_\wP )\not \in \nonegset$.
\end{proof}

\section*{Acknowledgment}

We thank St{\'e}phane Gaubert and Laurent Hardouin for the fruitful discussions on the topics covered in this paper.
\zor{We also thank the anonymous reviewers for their suggestions, which helped to improve an earlier version of the manuscript.}

This work was funded by the Deutsche Forschungsgemeinschaft (DFG, German Research Foundation), Projektnummer RA 516/14-1.
Partially supported by Deutsche Forschungsgemeinschaft (DFG, German Research Foundation) under Germany's Excellence Strategy -- EXC 2002/1 "Science of Intelligence" -- project number 390523135.

\section*{Conflict of interest}

The authors have no conflicts of interest to declare that are relevant to this article.

\bibliography{references.bib}

\begin{thebibliography}{10}
\providecommand{\doi}[1]{\url{https://doi.org/#1}}
\bibcommenthead

\bibitem{5628259}
Declerck P.
\newblock From Extremal Trajectories to Token Deaths in {P}-time Event Graphs.
\newblock IEEE Transactions on Automatic Control. 2011;56(2):463--467.
\newblock \doi{10.1109/TAC.2010.2091297}.

\bibitem{ZORZENON202219}
Zorzenon D, Balun J, Raisch J.
\newblock Weak Consistency of {P}-time Event Graphs.
\newblock IFAC-PapersOnLine. 2022;55(40):19--24.
\newblock 1st IFAC Workshop on Control of Complex Systems COSY 2022.
  \doi{https://doi.org/10.1016/j.ifacol.2023.01.042}.

\bibitem{iteb2006control}
Ouerghi I, Hardouin L.
\newblock Control synthesis for {P}-temporal event graphs.
\newblock In: 2006 8th International Workshop on Discrete Event Systems. IEEE;
  2006. p. 229--234.

\bibitem{brunsch2014modeling}
Brunsch T.: Modeling and control of complex systems in a dioid framework. PhD
  thesis, LARIS -- Université d'Angers. TU Berlin. 2014.

\bibitem{komenda2011application}
Komenda J, Lahaye S, {\v{S}}pa{\v{c}}ek P.
\newblock Application of product dioids for dead token detection in interval
  {P}-time event graphs.
\newblock IFAC Proceedings Volumes. 2011;44(1):6054--6059.

\bibitem{zorzenon2020bounded}
Zorzenon D, Komenda J, Raisch J.
\newblock Bounded Consistency of {P}-Time Event Graphs.
\newblock In: 2020 59th IEEE Conference on Decision and Control (CDC); 2020. p.
  79--85.

\bibitem{vspavcek2021analysis}
{\v{S}}pa{\v{c}}ek P, Komenda J, Lahaye S.
\newblock Analysis of {P}-time event graphs in (max,+) and (min,+) semirings.
\newblock International Journal of Systems Science. 2021;52(4):694--709.

\bibitem{lee2005extended}
Lee TE, Park SH.
\newblock An extended event graph with negative places and tokens for time
  window constraints.
\newblock IEEE Transactions on Automation Science and Engineering.
  2005;2(4):319--332.

\bibitem{munier2011graph}
Munier~Kordon A.
\newblock A graph-based analysis of the cyclic scheduling problem with time
  constraints: schedulability and periodicity of the earliest schedule.
\newblock Journal of Scheduling. 2011;14(1):103--117.

\bibitem{zorzenon2024consistency}
Zorzenon D, Raisch J.
\newblock Consistency of P-time event graphs is decidable in polynomial time.
\newblock IFAC-PapersOnLine. 2024;58(1):54--59.

\bibitem{gallai1958maximum}
Gallai T.
\newblock Maximum-minimum {S}{\"a}tze {\"u}ber {G}raphen.
\newblock Acta Mathematica Hungarica. 1958;9(3-4):395--434.

\bibitem{SCHRIJVER20051}
Schrijver A.
\newblock On the History of Combinatorial Optimization (Till 1960).
\newblock In: Aardal K, Nemhauser GL, Weismantel R, editors. Discrete
  Optimization. vol.~12 of Handbooks in Operations Research and Management
  Science. Elsevier; 2005. p. 1--68.
\newblock Available from:
  \url{https://www.sciencedirect.com/science/article/pii/S0927050705120015}.

\bibitem{baccelli1992synchronization}
Baccelli F, Cohen G, Olsder GJ, Quadrat JP.
\newblock Synchronization and linearity: an algebra for discrete event systems.
\newblock John Wiley \& Sons Ltd; 1992.

\bibitem{singer2003some}
Singer I.
\newblock Some Relations Between Linear Mappings and Conjugations in Idempotent
  Analysis.
\newblock Journal of Mathematical Sciences. 2003;115(5):2610--2630.

\bibitem{droste2009handbook}
Droste M, Kuich W, Vogler H.
\newblock Handbook of weighted automata.
\newblock Springer Science \& Business Media; 2009.

\bibitem{hopcroft1979introduction}
Hopcroft JE, Ullman JD.
\newblock Introduction to automata theory, languages, and computation.
\newblock 1st ed. Addison-Wesley; 1979.

\bibitem{hoefting1995minimum}
H\"{o}fting F, Wanke E.
\newblock Minimum Cost Paths in Periodic Graphs.
\newblock SIAM Journal on Computing. 1995;24(5):1051--1067.
\newblock \doi{10.1137/S0097539792234378}.

\bibitem{orlin1984some}
Orlin JB.
\newblock Some Problems on Dynamic/Periodic Graphs.
\newblock In: Pulleyblank WR, editor. Progress in Combinatorial Optimization.
  Academic Press; 1984. p. 273--293.

\bibitem{declerck2009extremal}
Declerck P.
\newblock From extremal trajectories to consistency in P-time event graphs;
  2009.
\newblock Technical report.

\bibitem{cormen2022introduction}
Cormen TH, Leiserson CE, Rivest RL, Stein C.
\newblock Introduction to Algorithms, fourth edition.
\newblock MIT Press; 2022.

\bibitem{khansa1996p}
Khansa W, Denat JP, Collart-Dutilleul S.
\newblock P-time {P}etri nets for manufacturing systems.
\newblock In: International Workshop on Discrete Event Systems, WODES. vol.~96;
  1996. p. 94--102.

\bibitem{zorzenon2023switched}
Zorzenon D, Komenda J, Raisch J.
\newblock Switched max-plus linear-dual inequalities: cycle time analysis and
  applications.
\newblock Discrete Event Dynamic Systems. 2024;34(1):199--250.

\bibitem{paek2020analysis}
{\v{S}}pa{\v{c}}ek P, Komenda J, Lahaye S.
\newblock {Analysis of P-time event graphs in (max,+) and (min,+) semirings}.
\newblock International Journal of Systems Science. 2021;0(0):1--16.
\newblock \doi{10.1080/00207721.2020.1837992}.
\newblock
  {\href{https://arxiv.org/abs/https://doi.org/10.1080/00207721.2020.1837992}{{https://doi.org/10.1080/00207721.2020.1837992}}}.

\bibitem{munier1996basic}
Munier A.
\newblock The basic cyclic scheduling problem with linear precedence
  constraints.
\newblock Discrete applied mathematics. 1996;64(3):219--238.

\bibitem{gaubert2023personal}
Gaubert S.: Personal communication. 2023.

\bibitem{haase2014integer}
Haase C, Halfon S.
\newblock Integer vector addition systems with states.
\newblock In: International Workshop on Reachability Problems. Springer; 2014.
  p. 112--124.

\bibitem{czerwinski2022reachability}
Czerwi{\'n}ski W, Orlikowski {\L}.
\newblock Reachability in vector addition systems is Ackermann-complete.
\newblock In: 2021 IEEE 62nd Annual Symposium on Foundations of Computer
  Science (FOCS). IEEE; 2022. p. 1229--1240.

\end{thebibliography}

\end{document}